\documentclass[12pt]{article}

\usepackage{amsmath, amssymb, enumerate, amsthm}
\usepackage{threeparttable}
\usepackage{algorithm}
\usepackage{algpseudocode}
\usepackage{natbib}

\usepackage{mathrsfs}
\usepackage{booktabs}
\usepackage{multirow}
\usepackage{rotating}
\usepackage{bm}
\usepackage{caption}
\usepackage{graphicx}
\usepackage{epstopdf}

\newtheorem{theorem}{Theorem}
\newtheorem{lemma}[theorem]{Lemma}
\newtheorem{remark}{Remark}

\usepackage[left=1in,right=1in,top=1in,bottom=1in]{geometry}
\date{}

\begin{document}

	\def\spacingset#1{\renewcommand{\baselinestretch}%
		{#1}\small\normalsize} \spacingset{1}

	\title{\bf {Empirical likelihood inference for longitudinal data with covariate measurement errors:\\ An application to the LEAN study} }

	\author{Yuexia Zhang \\
		Department of Computer and Mathematical Sciences, University of Toronto\\
		Guoyou Qin\\
		Department of Biostatistics, Fudan University\\
		Zhongyi Zhu\\
		Department of Statistics, Fudan University\\
		and\\
		Jiajia Zhang\\
		Department of Epidemiology and Biostatistics, University of South Carolina
	}
	\maketitle
	
	\begin{abstract}
Measurement errors usually arise during the longitudinal data collection process. Ignoring the effects of measurement errors will lead to invalid estimates. The Lifestyle Education for Activity and Nutrition (LEAN) study was designed to assess the effectiveness of intervention for enhancing weight loss over nine months. The covariates systolic blood pressure (SBP) and diastolic blood pressure (DBP) were measured at baseline, month $4$, and month $9$. At each assessment time, there were two replicate measurements for SBP and DBP. The replicate measurement errors of SBP follow different distributions, as does DBP. To account for the distributional difference of replicate measurement errors, a new method for analyzing longitudinal data with replicate covariate measurement errors is developed based on the empirical likelihood method. The asymptotic properties of the proposed estimator are established under some regularity conditions. The confidence region for the parameters of interest can be constructed based on the chi-squared approximation without estimating the covariance matrix. Additionally, the proposed empirical likelihood estimator is asymptotically more efficient than the estimator of  \cite{Lin2018Analysis}. Extensive simulations demonstrate that the proposed method can eliminate the effects of measurement errors in the covariate and has a high estimation efficiency. The proposed method indicates the significant effect of the intervention on BMI in the LEAN study
		
	\end{abstract}
	
	\noindent%
	\vfill

\spacingset{1.5} % DON'T change the spacing!

\section{Introduction}
\label{s:intro}

Longitudinal data are commonly seen in various fields, such as psychology, economics, social sciences, and public health, and measurement errors usually arise during the data collection process. The Lifestyle Education for Activity and Nutrition (LEAN) study \citep{Barry2011Using} was designed to assess the effectiveness of  intervention for enhancing weight loss over nine months in sedentary overweight or obese adults. In this study, $197$  men and women between the ages of 18 and 64 who were underactive, overweight, or obese (BMI $\geq 25$), and had access to the internet were randomly assigned to the standard care group and the intervention group. For each participant, systolic blood pressure (SBP) and diastolic blood pressure (DBP) were measured at baseline, month $4$, and month $9$. As pointed out by \cite{qin2016simultaneous,Qin2016Robust} and  \cite{Lin2018Analysis}, there exist measurement errors in the covariates SBP and DBP.

We denote the surrogate values of SBP as ${\rm SBP}_{(1)}$ and ${\rm SBP}_{(2)}$, and the corresponding measurement errors as $\xi_{(1)}$ and $\xi_{(2)}$. If we further assume the additive measurement error models for ${\rm SBP}_{(1)}$ and ${\rm SBP}_{(2)}$, then ${\rm cSBP}_{(1)}\triangleq {\rm SBP}_{(1)}-({\rm SBP}_{(1)}+{\rm SBP}_{(2)})/2=(\xi_{(1)}-\xi_{(2)})/2$.  If $\xi_{(1)}$ and $\xi_{(2)}$ follow the same distribution, then the density function of ${\rm cSBP}_{(1)}$ is symmetric. Similarly, we denote one of the centralized surrogate values of DBP as ${\rm cDBP}_{(1)}$. Figure~\ref{figure:1} displays the density functions of ${\rm cSBP}_{(1)}$ and ${\rm cDBP}_{(1)}$, which illustrates that the density functions of ${\rm cSBP}_{(1)}$ and ${\rm cDBP}_{(1)}$ are not symmetric. We further find that the density functions of ${\rm cSBP}_{(1)}$ and ${\rm cDBP}_{(1)}$ are significantly asymmetric at the significance level of $0.05$ based on  the D'Agostino skewness test statistic  \citep{d1970transformation}. Therefore, the replicate measurement errors of SBP  follow different distributions, as does DBP. However, few existing methods  have accounted for this distributional difference in measurement errors. The main purpose of this paper is to develop a new longitudinal data analysis method which can account for the  distributional difference of replicate measurement errors.

\begin{figure}
	% FIGURE 1
	\centerline{\includegraphics[width=35em]{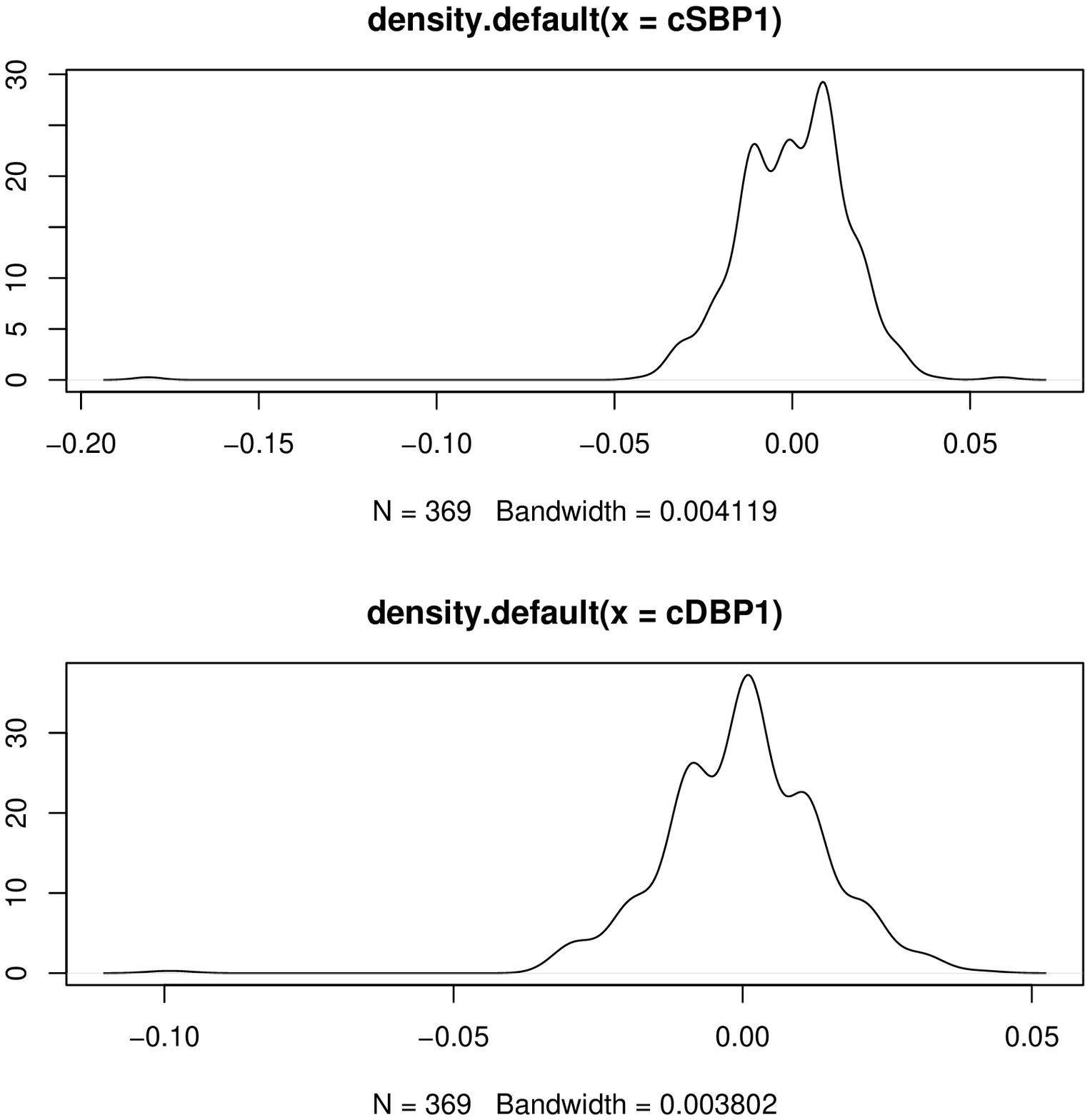}}
	\caption{Density functions of ${\rm cSBP}_{(1)}$ and ${\rm cDBP}_{(1)}$.
		\label{figure:1}}
\end{figure}	

The likelihood-based method and estimating equation method are the most popular methods in longitudinal data analysis \citep{laird1982random,Liang1986Longitudinal,diggle2002analysis,zhang2015joint,cheng2016efficient,funatogawa2018longitudinal}. The likelihood-based method is generally efficient but sensitive to the distribution misspecification, because it always assumes the joint distribution of repeated observations for each subject and applies the maximum likelihood estimation (MLE) method or restricted maximum likelihood estimation (REML) method to estimate. The estimating equation method avoids making assumptions for the multivariate distribution by specifying the first two moments of response. However, the estimating equation method cannot deal with the problem where the number of estimating functions is larger than the number of parameters. The empirical likelihood method \citep{owen1988empirical}, a combination of the likelihood-based method and estimating equation method, has attracted much attention recently
\citep{wang2010generalized,qiu2015moving,zhao2019new,hu2022efficient}.
The empirical likelihood method is nonparametric, distribution-free, and also enjoys some good properties of the parametric likelihood method. For example, the empirical likelihood ratio statistic  asymptotically follows a chi-squared distribution \citep{owen1990empirical,owen2001empirical}.
At the same time, it can deal with the problem where there are more estimating functions than parameters. Besides, it can combine the information in the estimating functions in a most efficient way \citep{Qin1994Empirical}.

Although there is  considerable literature on how to deal with measurement errors  using the likelihood-based method and estimating equation method \citep{wulfsohn1997joint, Wang2000Expected, Wu2002A, hsieh2006joint, Wang2006Corrected, qin2016simultaneous, li2019semiparametric},  the empirical likelihood method is not widely used in analysing longitudinal data with measurement errors. \cite{zhao2009empirical} investigated the empirical likelihood inference for semiparametric varying-coefficient partially linear error-in-variables models, where they applied the correction for attenuation technique to construct a bias-corrected auxiliary random vector. However, their method needs to make some assumptions about the covariance matrix of measurement errors, which may not be satisfied in practice.

For the LEAN study,  \cite{Lin2018Analysis} constructed an unbiased estimating equation using the independence between replicate measurement errors to eliminate the effects of measurement errors. Although their method is asymptotically more efficient than the method of \cite{qin2016simultaneous}, their method may lose some efficiency if the distributions of replicate measurement errors are different.
Therefore,  it is important to consider the distributional difference of measurement  errors  to get a more efficient estimator. In this paper, we propose a new empirical likelihood-based method. The proposed  estimator is  asymptotically more efficient than the estimator of \cite{Lin2018Analysis}. In addition, the proposed method can deal with the problem where there are more than two replicate measurements at each assessment time.

The remainder of this paper is organized as follows. In Section~\ref{s:model}, we introduce the mean model and measurement error process. The proposed empirical likelihood-based  method is outlined in  Section~\ref{s:proposedmethod} and some asymptotic properties are established in Section~\ref{s:asympro}. We assess the performance of the proposed method with simulation studies in Section~\ref{s:simu} and apply the proposed method to the LEAN data set in Section~\ref{s:realdata}. The paper is concluded with a discussion in Section~\ref{s:discussion}. The detailed proofs are given in the appendices. The R codes for  simulation studies are available on the RunMyCode website.

\section{Model specification}
\label{s:model}

\subsection{Mean model}
\label{ss:meanmodel}
In this paper, we consider a longitudinal  study  with $n$ subjects and  $m_{i}$ observations over time  for the $i$th  subject. Let $Y_{ij}$ be  the response  variable and $\bm{X}_{ij}=(X_{ij,1},\cdots,X_{ij,p})^\top$ be the vector of covariates, where $i=1,\cdots,n$, $j=1,\cdots,m_{i}$, and $\{m_{i}, i=1,\cdots,n\}$ are bounded positive integers.
Assume the  longitudinal data set follows the  linear regression model, i.e.,
\begin{equation}
	Y_{ij}=\bm{X}_{ij}^\top\bm{\beta}_{0}+\varepsilon_{ij},\quad i=1,\cdots,n, \quad j=1,\cdots,m_{i},
	\label{eq:lrm}
\end{equation}
where   $\bm{\beta}_{0}=(\beta_{01},\cdots,\beta_{0p})^\top$ is a $p$-dimensional  unknown  vector and $\varepsilon_{ij}$ is the random error term.  In matrix form,  we denote  $\bm{Y}_{i}=(Y_{i1},\cdots,Y_{im_{i}})^\top$, $\mathbf{X}_{i}=(\bm{X}_{i1},\cdots,\bm{X}_{im_{i}})^\top$, and $\bm{\varepsilon}_{i}=(\varepsilon_{i1},\cdots,\varepsilon_{im_{i}})^\top$, where $\{\bm{\varepsilon}_{i},i=1,\ldots,n\}$ are mutually independent with  $\mathrm{E}(\bm{\varepsilon}_{i}|\mathbf{X}_{i})=\bm{0}$ and  covariance matrix $\bm{\Sigma}_{i}$ for each $i\in\{1,\ldots,n\}$.

\subsection{Measurement error process}
\label{ss:measuremodel}
Let $\bm{X}_{ij}$ denote the covariate vector measured with error, and $\bm{W}_{ij}$ denote its observed version. Then, we assume that $\bm{X}_{ij}$ and $\bm{W}_{ij}$ follow a classical additive measurement error model, i.e.,
\begin{equation*}
	\bm{W}_{ij}=\bm{X}_{ij}+\bm{\xi}_{ij},
\end{equation*}
where $\bm{\xi}_{ij}$ is the measurement error with mean zero, and $\bm{\xi}_{ij}$ is independent of $\mathbf{X}_{i}$ and $\bm{\varepsilon}_{i}$.

In practice, replicate measurements of $\bm{X}_{ij}$ are often conducted  to get more reliable results. We assume that there exist $K$ $(K\geq 2)$ replicate measurements for the error-prone covariate $\bm{X}_{ij}$, i.e.,
\[
\bm{W}_{ij(k)}=\bm{X}_{ij}+\bm{\xi}_{ij(k)}, \quad k=1,\cdots,K,
\]
where $\{\bm{\xi}_{ij(k)}, k=1,\cdots,K\}$ are mutually independent and the distributions of  measurement errors $\{\bm{\xi}_{ij(k)}, k=1,\cdots,K\}$ can be different.
For  convenience, we denote $\mathbf{W}_{i(k)}=(\bm{W}_{i1(k)},\cdots,\bm{W}_{im_{i}(k)})^\top$, $\bm{W}_{ij(k)}=(W_{ij(k),1},\cdots,W_{ij(k),p})^\top$,  and $\bm{\xi}_{i(k)}=(\bm{\xi}_{i1(k)},\cdots,\bm{\xi}_{im_{i}(k)})^\top$ for $k=1,\cdots,K$.

\section{Proposed method}
\label{s:proposedmethod}
When there are two replicate measurements for $\bm{X}_{ij}$,  \cite{Lin2018Analysis} proposed the following estimating equation for estimation of $\bm{\beta}_{0}$
\[
\sum_{i=1}^{n}\left\{\mathbf{W}_{i(1)}^\top\bm{\Sigma}_{i}^{-1}(\bm{Y}_{i}-\mathbf{W}_{i(2)}\bm{\beta})\!+\!\mathbf{W}_{i(2)}^\top\bm{\Sigma}_{i}^{-1}(\bm{Y}_{i}-\mathbf{W}_{i(1)}\bm{\beta})\right\}=\bm{0}.
\]
When there are more than two measurements, we can extend \cite{Lin2018Analysis}'s method directly and use the following estimating equation 
\begin{equation}
\sum_{i=1}^{n}\sum_{k_{1}\neq k_{2}}\mathbf{W}_{i(k_{1})}^\top\bm{\Sigma}_{i}^{-1}(\bm{Y}_{i}-\mathbf{W}_{i(k_{2})}\bm{\beta})=\bm{0},
\label{eq:multiplemeasure}
\end{equation}
where $k_{1},k_{2}\in\{1,\ldots,K\}$, and $K>2$. Denote the solution to (\ref{eq:multiplemeasure}) as $\hat{\bm{\beta}}_{LIN}$.

However, the estimating equation  \eqref{eq:multiplemeasure} does not consider the heterogeneity of  different measurement errors, because it gives the same weight to all the estimating functions  $\sum_{i=1}^{n}\mathbf{W}_{i(k_{1})}^\top\bm{\Sigma}_{i}^{-1}(\bm{Y}_{i}-\mathbf{W}_{i(k_{2})}\bm{\beta})$, where $k_{1}\neq k_{2}$. As a result, $\hat{\bm{\beta}}_{LIN}$ may not be highly efficient. To improve the estimation efficiency, we propose  to estimate  $\bm{\beta}_{0}$ based on the empirical likelihood method.

First, we introduce an auxiliary random vector  as follows
\begin{equation}
	\begin{aligned}
		\bm{g}_{i}(\bm{\beta})&=\begin{pmatrix}\mathbf{W}_{i(1)}^\top\bm{\Sigma}_{i}^{-1}(\bm{Y}_{i}-\mathbf{W}_{i(2)}\bm{\beta}) \\ \mathbf{W}_{i(2)}^\top\bm{\Sigma}_{i}^{-1}(\bm{Y}_{i}-\mathbf{W}_{i(1)}\bm{\beta}) \\ \vdots\\ \mathbf{W}_{i(K-1)}^\top\bm{\Sigma}_{i}^{-1}(\bm{Y}_{i}-\mathbf{W}_{i(K)}\bm{\beta}) \\  \mathbf{W}_{i(K)}^\top\bm{\Sigma}_{i}^{-1}(\bm{Y}_{i}-\mathbf{W}_{i(K-1)}\bm{\beta}) \end{pmatrix}.
		\label{eq:auxivector}
	\end{aligned}
\end{equation}
Because of the independence between replicate measurement errors, the auxiliary random vector has expectation zero if $\bm{\beta}=\bm{\beta}_{0}$. Thus, the effects of measurement errors can be eliminated. However, the elements in  $\bm{g}_{i}(\bm{\beta})$ are not functionally independent in all situations. As illustrated in  a toy example in \ref{ss:construction}, there exist some duplicate elements in $\bm{g}_{i}(\bm{\beta})$. Besides, there is an inner relationship among the elements of $\bm{g}_{i}(\bm{\beta})$. Both factors make the matrix  $\mathrm{E}\{\bm{g}_{i}(\bm{\beta}_{0})\bm{g}_{i}(\bm{\beta}_{0})^\top\}$ not positive definite. However, positive definiteness of  the matrix  $\mathrm{E}\{\bm{g}_{i}(\bm{\beta}_{0})\bm{g}_{i}(\bm{\beta}_{0})^\top\}$ is one of the necessary conditions for the asymptotic normality of the empirical likelihood estimator, as shown in Section~\ref{s:asympro}. Therefore, we need to eliminate the elements   which are functionally dependent or have inner relationships with other elements in $\bm{g}_{i}(\bm{\beta})$ from  $\bm{g}_{i}(\bm{\beta})$. Denote the reduced auxiliary  random vector as $\bm{g}_{i}^{*}(\bm{\beta})$ and assume the dimension of  $\bm{g}_{i}^{*}(\bm{\beta})$ is $q$, where $\bm{g}_{i}^{*}(\bm{\beta})$ satisfies the condition that $\mathrm{E}\{\bm{g}_{i}^{*}(\bm{\beta}_{0})\bm{g}_{i}^{*}(\bm{\beta}_{0})^\top\}$ is a positive definite  matrix. We illustrate how to obtain the reduced auxiliary random vector  by using the toy example, which is provided in \ref{ss:construction}. In general,  the reduced  auxiliary random vector $\bm{g}_{i}^{*}(\bm{\beta})$ can be obtained
based on Algorithm \ref{algorithm1}.

 \begin{algorithm}
	\caption{The proposed procedure for obtaining  the reduced  auxiliary random vector $\bm{g}_{i}^{*}(\bm{\beta})$}
	\label{algorithm1}
	\begin{algorithmic}
		\item [1.] Write the complete formula of $\bm{g}_{i}(\bm{\beta})$ based on \eqref{eq:auxivector}.
		\item [2.] Check whether there are some duplicate elements in $\bm{g}_{i}(\bm{\beta})$. If there are some  duplicate elements, then keep the unique elements and eliminate the duplicate elements from $\bm{g}_{i}(\bm{\beta})$. Denote the reduced random vector as $\tilde{\bm{g}}_{i}(\bm{\beta})$; if there is no duplicate element, then let $\tilde{\bm{g}}_{i}(\bm{\beta})=\bm{g}_{i}(\bm{\beta})$.
		\item [3.] Write the complete formula of $\mathrm{E}\{\tilde{\bm{g}}_{i}(\bm{\beta}_{0})\tilde{\bm{g}}_{i}(\bm{\beta}_{0})^\top\}$ based on  model assumptions.
		\item [4.] Check whether there are some elements in $\mathrm{E}\{\tilde{\bm{g}}_{i}(\bm{\beta}_{0})\tilde{\bm{g}}_{i}(\bm{\beta}_{0})^\top\}$ that can be  represented as a linear function of  other elements in $\mathrm{E}\{\tilde{\bm{g}}_{i}(\bm{\beta}_{0})\tilde{\bm{g}}_{i}(\bm{\beta}_{0})^\top\}$. If it is true, then eliminate the corresponding elements in $\tilde{\bm{g}}_{i}(\bm{\beta})$ from $\tilde{\bm{g}}_{i}(\bm{\beta})$. Denote the reduced auxiliary random vector as $\bm{g}_{i}^{*}(\bm{\beta})$; if it is false, then let $\bm{g}_{i}^{*}(\bm{\beta})=\tilde{\bm{g}}_{i}(\bm{\beta})$.		
	\end{algorithmic}
\end{algorithm}

Second, following the standard procedure for the empirical likelihood method, we define the profile empirical likelihood ratio function as
\begin{equation}
	R(\bm{\beta})=\max\left\{\prod_{i=1}^{n}(n\pi_{i})\Bigm|\pi_{i} \geq 0, \sum_{i=1}^{n}\pi_{i}=1, \sum_{i=1}^{n}\pi_{i}\bm{g}_{i}^{*}(\bm{\beta})=\bm{0} \right\}.
	\label{eq:profileratio}
\end{equation}
Using the Lagrange multiplier method, $R(\bm{\beta})$  is maximized at
\begin{equation}
	\pi_{i}=\frac{1}{n\big\{1+\bm{\lambda}^\top\bm{g}_{i}^{*}(\bm{\beta})\big\}}, \quad i=1,\cdots,n,
	\label{eq:weight}
\end{equation}
where the Lagrange multiplier $\bm{\lambda}=(\lambda_{1},\cdots,\lambda_{q})^\top$ satisfies the following condition
\begin{equation}
	\frac{1}{n}\sum_{i=1}^{n}\frac{\bm{g}_{i}^{*}(\bm{\beta})}{1+\bm{\lambda}^\top\bm{g}_{i}^{*}(\bm{\beta})}=\bm{0}.
	\label{eq:constraint}
\end{equation}
Based on~(\ref{eq:profileratio}) and~(\ref{eq:weight}), we have
\begin{equation}
	-2\log R(\bm{\beta})=-2\log\left[\prod_{i=1}^{n}\big\{1+\bm{\lambda}^\top\bm{g}_{i}^{*}(\bm{\beta})\big\}^{-1}\right]=2\sum_{i=1}^{n}\log\big\{1+\bm{\lambda}^\top\bm{g}_{i}^{*}(\bm{\beta})\big\}.
	\label{eq:object}
\end{equation}
The maximum empirical likelihood estimator (MELE) of $\bm{\beta}_{0}$, $\hat{\bm{\beta}}$, can be obtained by maximizing $R(\bm{\beta})$ or minimizing $-2\log R(\bm{\beta})$ under the constraint~(\ref{eq:constraint}).

In general, we can estimate $\bm{\beta}_{0}$ based on Algorithm \ref{algorithm2}. 

 \begin{algorithm}
	\caption{The proposed procedure for  estimating $\bm{\beta}_{0}$}
	\label{algorithm2}
	\begin{algorithmic}
	\item [1.]  Choose an initial value ${\bm{\beta} }^{\left(0\right)}$, which can be obtained by using the method in  \cite{Lin2018Analysis} with a working independence correlation matrix. Set $k=0$.
\item [2.] With the value of ${\bm{\beta} }^{\left(k\right)}$, estimate the covariance matrix $\bm{\Sigma}_{i}$ by using the same method as that in  \cite{Qin2016Robust}. Denote the estimated value of  covariance matrix as $\hat{\bm{\Sigma}}_{i}^{(k)}$.
\item [3.] Construct the reduced auxiliary random vector $\bm{g}_{i}^{*}({\bm{\beta} }^{\left(k\right)},\hat{\bm{\Sigma}}_{i}^{(k)})$ based on Algorithm \ref{algorithm1}, and solve equation~(\ref{eq:constraint}) to obtain $\bm{\lambda}^{\left(k\right)}$ by using the modified Newton-Raphson method. Take the value of $\bm{\lambda}^{\left(k\right)}$ into the objective function~(\ref{eq:object}) and obtain the value of $-2\log R({\bm{\beta} }^{\left(k\right)})$. Then  calculate the new estimated value  ${\bm{\beta} }^{\left(k+1\right)}$  based on an optimization method (e.g., the Broyden-Fletcher-Goldfarb-Shanno (BFGS) algorithm).  Set $k=k+1$.
\item [4.] Iterate Step 2 and Step 3 until convergence. The final estimated value of $\bm{\beta}_{0}$ is denoted as  $\hat{\bm{\beta}}$.
	\end{algorithmic}
\end{algorithm}

\section{Asymptotic properties}
\label{s:asympro}
This section shows the asymptotic properties of the proposed estimator.
Specially, Theorem~\ref{theo:normality} presents the asymptotic normality of the proposed estimator and Theorem~\ref{theo:efficiency} shows that the proposed estimator is asymptotically more efficient than the estimator of \cite{Lin2018Analysis}. Theorems~\ref{theo:chisquaretrue},~\ref{theorem:fulltest} and~\ref{theorem:profiletest} show the properties of statistics, which are obtained from the empirical likelihood ratio function.  To establish the  asymptotic properties, we introduce the following regularity conditions:

\begin{enumerate}
	\item[(R.1)]  The number of replicate measurements for the error-prone covariate $\bm{X}_{ij}$, $K$, is a bounded positive integer.
	\item[(R.2)]   The regression parameter $\bm{\beta}_0$ is identifiable, i.e., there is a unique $\bm{\beta}_{0}\in \mathscr{B}$ satisfying the model assumption~(\ref{eq:lrm}) which guarantees $\mathrm{E}\{\bm{g}_{i}^{*}(\bm{\beta}_{0})\}=\bm{0}$, where  $\mathscr{B}$ is a compact parameter space.
	\item[(R.3)] There exist two positive constants $c_{1}$ and $c_{2}$ such that
	\[
	0<c_{1}\leq \min_{1\leq i\leq n}\eta_{i1}\leq \max_{1\leq i\leq n}\eta_{im_{i}}\leq c_{2}<\infty,
	\]
	where $\eta_{i1}$ and $\eta_{im_{i}}$ denote the smallest and largest eigenvalues of $\bm{\Sigma}_{i}$, respectively.
	\item[(R.4)] $\max_{1\leq i\leq n}\mathrm{E}\|\mathbf{X}_{i}\|^{6}<\infty$, $\max_{1\leq i\leq n,1\leq k\leq K}\mathrm{E}\|\bm{\xi}_{i(k)}\|^{3}<\infty$, and\\ $\max_{1\leq i\leq n}\sup_{\mathbf{x}}\mathrm{E}(\|\bm{\varepsilon}_{i}\|^{3}|\mathbf{X}_{i}=\mathbf{x})<\infty$, where $\|\cdot\|$ denotes the Euclidean norm.
	\item[(R.5)]  $\mathbf{L}_{n}/n\rightarrow \mathbf{L}$ in probability for some   matrix $\mathbf{L}$ and $\mathbf{M}_{n}/n\rightarrow \mathbf{M}$ in probability for some positive definite matrix  $\mathbf{M}$, where $\mathbf{L}_{n}=\sum_{i=1}^{n}\partial \bm{g}_{i}^{*}(\bm{\beta}_{0})/\partial \bm{\beta}^\top$  and $\mathbf{M}_{n}=\sum_{i=1}^{n}\bm{g}_{i}^{*}(\bm{\beta}_{0})\bm{g}_{i}^{*}(\bm{\beta}_{0})^\top$.	
\end{enumerate}
\begin{remark}
Condition (R.1) requires the number of replicate measurements for $\bm{X}_{ij}$ to be bounded, it can ensure $\sum_{i=1}^{n}\mathrm{E}\|\bm{g}_{i}^{*}(\bm{\beta}_{0})/\sqrt{n}\|^{3}\rightarrow 0$. This condition is easy to verify in practice.  Condition (R.2) assumes the identifiability of the true parameter $\bm{\beta}_0$. It is not easy to check in practice, but it is a commonly used condition in empirical likelihood literature, see \cite{owen2001empirical}. Condition (R.3) requires the eigenvalues of  covariance matrices $\bm{\Sigma}_{i}$, $i=1,\ldots,n$ to be  bounded away from $0$ and $\infty$. If $\bm{\Sigma}_{i}$ $(i=1,\ldots,n)$  are known, Condition (R.3)  can be checked directly; if $\bm{\Sigma}_{i}$ $(i=1,\ldots,n)$  are unknown, we can first use some methods to estimate them, such as that in  \cite{Qin2016Robust}. Then we can check whether the sample covariance matrices satisfy  Condition (R.3)  or not. Sometimes, we assume the covariance matrices $\bm{\Sigma}_{i}$  $(i=1,\ldots,n)$  satisfy some specific structures, such as independent structure, exchangeable structure, or AR(1) structure, then Condition (R.3) can be satisfied naturally.  Condition (R.4) contains the moment conditions for the covariate, measurement error, and random error, which can be met under some common distributions, such as normal distribution and exponential distribution. Condition (R.5) assumes the convergence of $\mathbf{L}_{n}/n$ and $\mathbf{M}_{n}/n$  in probability. If the distribution of variables and true models are known, then Conditions (R.4)--(R.5) can be checked; otherwise, it is not easy to check. However, the conditions which are similar to (R.4)--(R.5) can be found be in many references, such as \cite{xue2007empiricallike} and \cite{zhang2019novel}.
\end{remark}
\begin{theorem}
	\label{theo:normality}
	Assuming that conditions (R.1)--(R.5) hold,  we have 
	\[
	\sqrt{n}(\hat{\bm{\beta}}-\bm{\beta}_{0})\rightsquigarrow \mathcal{N}\big(\bm{0}, (\mathbf{L}^\top\mathbf{M}^{-1}\mathbf{L})^{-1}\big).
	\]
\end{theorem}

When $\bm{g}_{i}^{*}(\bm{\beta})=\bm{g}_{i}(\bm{\beta})$, according to Theorem 3.6 in \cite{owen2001empirical}, the asymptotic variance of the empirical likelihood estimator $\hat{\bm{\beta}}$ is at least as small as that of $\hat{\bm{\beta}}_{LIN}$, because $\sum_{k_{1}\neq k_{2}}\mathbf{W}_{i(k_{1})}^\top\bm{\Sigma}_{i}^{-1}(\bm{Y}_{i}-\mathbf{W}_{i(k_{2})}\bm{\beta})$ is a linear combination of  $\bm{g}_{i}(\bm{\beta})$. When $\bm{g}_{i}^{*}(\bm{\beta})\neq\bm{g}_{i}(\bm{\beta})$, the above conclusion still holds because (1) the asymptotic variance of $\hat{\bm{\beta}}_{LIN}$ is equal to the  asymptotic variance of one certain estimator obtained from the estimating equation $\sum_{i=1}^{n}\mathbf{B}\bm{g}_{i}^{*}(\bm{\beta})=\bm{0}$, where $\mathbf{B}$ is a special $p\times q$ matrix and (2) the asymptotic variance of the empirical likelihood estimator $\hat{\bm{\beta}}$ is at least as small as that of any estimator obtained from the estimating equation $\sum_{i=1}^{n}\mathbf{C}\bm{g}_{i}^{*}(\bm{\beta})=\bm{0}$, where $\mathbf{C}$ is an arbitrary $p\times q$ matrix.  The proof of this conclusion in the toy example is given in \ref{ss:efficiency}, which can be easily extended to other cases. We summarize the conclusion in the following theorem.

\begin{theorem}
	\label{theo:efficiency}
	The asymptotic variance of the empirical likelihood estimator $\hat{\bm{\beta}}$ is at least as small as that of $\hat{\bm{\beta}}_{LIN}$. In other words, the  empirical likelihood estimator $\hat{\bm{\beta}}$ is at least as efficient as $\hat{\bm{\beta}}_{LIN}$.
\end{theorem}

In particular, if the replicate measurements for the covariate follow the same distribution or more generally, the models meet some specific moment conditions, then the asymptotic variance of  $\hat{\bm{\beta}}$ is the same as that of  $\hat{\bm{\beta}}_{LIN}$. For example, if there are two replicate measurements, the  moment condition is
\begin{multline*}
	\sum_{i=1}^{n}\big[\mathrm{E}\{\mathbf{X}_{i}^\top\Sigma_{i}^{-1}\mathrm{cov}(\bm{\delta}_{i(2)}\bm{\beta})\Sigma_{i}^{-1}\mathbf{X}_{i}\}+\mathrm{E}(\bm{\delta}_{i(1)}^\top\Sigma_{i}^{-1}\bm{\delta}_{i(1)})\\
	+\mathrm{cov}(\bm{\delta}_{i(1)}^\top\Sigma_{i}^{-1}\bm{\delta}_{i(2)}\bm{\beta})-\mathrm{E}(\bm{\delta}_{i(1)}^\top\Sigma_{i}^{-1}\bm{\delta}_{i(2)}\bm{\beta}\bm{\beta}^\top\bm{\delta}_{i(1)}^\top\Sigma_{i}^{-1}\bm{\delta}_{i(2)})\big]\\
	=\sum_{i=1}^{n}\big[\mathrm{E}\{\mathbf{X}_{i}^\top\Sigma_{i}^{-1}\mathrm{cov}(\bm{\delta}_{i(1)}\bm{\beta})\Sigma_{i}^{-1}\mathbf{X}_{i}\}+\mathrm{E}(\bm{\delta}_{i(2)}^\top\Sigma_{i}^{-1}\bm{\delta}_{i(2)})\\
	+\mathrm{cov}(\bm{\delta}_{i(2)}^\top\Sigma_{i}^{-1}\bm{\delta}_{i(1)}\bm{\beta})-\mathrm{E}(\bm{\delta}_{i(2)}^\top\Sigma_{i}^{-1}\bm{\delta}_{i(1)}\bm{\beta}\bm{\beta}^\top\bm{\delta}_{i(2)}^\top\Sigma_{i}^{-1}\bm{\delta}_{i(1)})\big].
\end{multline*}
The detailed proof is given in  \ref{ss:Comparison}. However, if the moment condition is violated, the asymptotic variance of $\hat{\bm{\beta}}$ is  smaller than that of  $\hat{\bm{\beta}}_{LIN}$. In summary, the proposed  empirical likelihood estimator is asymptotically more efficient than  \cite{Lin2018Analysis}'s estimator.

\begin{theorem}
	\label{theo:chisquaretrue}
	Assuming that conditions (R.1)--(R.5) hold,  we have $-2\log R(\bm{\beta}_{0})\rightsquigarrow \chi^{2}(q)$, where $\chi^{2}(q)$ is a chi-squared distribution with $q$ degrees of freedom.
\end{theorem}

\begin{theorem}
	\label{theorem:fulltest}
	The empirical likelihood ratio statistic for  the test of $H_{0}: \bm{\beta}=\bm{\beta}_{0}$ versus $H_{a}: \bm{\beta}\neq \bm{\beta}_{0}$ is
	\[
	W_{1}(\bm{\beta}_{0})=-2\log\big\{R(\bm{\beta}_{0})/R(\hat{\bm{\beta}})\big\}.
	\]
	In addition, assuming that conditions (R.1)--(R.5) hold,  we have $W_{1}(\bm{\beta}_{0})\rightsquigarrow \chi^{2}(p)$  under $H_{0}$.
\end{theorem}

Similar to the parametric likelihood method, Theorems~\ref{theo:chisquaretrue} and~\ref{theorem:fulltest} allow us to use the test statistics $-2\log R(\bm{\beta}_{0})$ and $W_{1}(\bm{\beta}_{0})$ to perform  hypothesis testing and  construct confidence region for $\bm{\beta}_{0}$. Specifically, the $100(1-\alpha)\%$ empirical likelihood confidence region for $\bm{\beta}_{0}$ can be constructed as
\[
\mathbf{I}_{1}=\left\{\bm{\beta}:-2\log R(\bm{\beta})\leq \chi^{2}_{1-\alpha}(q)\right\}\quad  \mathrm{or} \quad \mathbf{I}_{2}=\left\{\bm{\beta}:W_{1}(\bm{\beta})\leq \chi^{2}_{1-\alpha}(p)\right\},
\]
where $\chi^{2}_{1-\alpha}(d)$ is the $(1-\alpha)$th quantile of $\chi^{2}(d)$ for any positive integer $d$.

If we are  only interested in a part of elements in the parameter vector $\bm{\beta}$, we can use a profile empirical likelihood ratio test statistic to perform  hypothesis testing and construct confidence region for the parameters of interest. Let $\bm{\beta}=(\bm{\beta}_{1}^\top,\bm{\beta}_{2}^\top)^\top$, where $\bm{\beta}_{1}$ and $\bm{\beta}_{2}$ are $r$-dimensional and $(p-r)$-dimensional vectors, respectively. In order to test $H_{0}: \bm{\beta}_{1}=\bm{\beta}_{1}^{0}$ versus $H_{a}: \bm{\beta}_{1}\neq \bm{\beta}_{1}^{0}$, we define the profile empirical likelihood ratio test statistic as
\[
W_{2}(\bm{\beta}_{1}^{0})=-2\log \big\{R(\bm{\beta}_{1}^{0},\hat{\bm{\beta}}_{2}^{0})/R(\hat{\bm{\beta}}_{1},\hat{\bm{\beta}}_{2})\big\},
\]
where $\hat{\bm{\beta}}_{2}^{0}$ minimizes $-2\log R(\bm{\beta}_{1}^{0},\bm{\beta}_{2})$ with respect to $\bm{\beta}_{2}$. Furthermore, we can show the following theorem:

\begin{theorem}
	\label{theorem:profiletest}
	Assuming that conditions (R.1)--(R.5) hold,  we have $W_{2}(\bm{\beta}_{1}^{0})\rightsquigarrow \chi^{2}(r)$ under $H_{0}$.
\end{theorem}
Thus, an approximate $100(1-\alpha)\%$ confidence region for $\bm{\beta}_{1}^{0}$  is
\[
\mathbf{I}_{3}=\left\{\bm{\beta}_{1}: W_{2}(\bm{\beta}_{1})\leq \chi^{2}_{1-\alpha}(r)\right\}.
\]

The proofs of Theorems~\ref{theo:normality},~\ref{theo:chisquaretrue},~\ref{theorem:fulltest}, and~\ref{theorem:profiletest} are given in~\ref{ss:proofofasym}.
When constructing the confidence region for the parameters of interest, compared with the normal approximation-based method, the chi-squared approximation-based method  can avoid estimating the asymptotic covariance matrix  and does not need to impose prior constraints on the shape of the confidence region. Therefore, the chi-squared approximation-based method is recommended instead of the normal approximation-based method in practice.

\section{Simulation studies}
\label{s:simu}
In this section, we conduct comprehensive simulations  to evaluate the performance of the proposed empirical likelihood (EL) method when there are replicate measurements for the covariate and there may exist measurement errors in the covariate. For comparison, we also present the simulation results of  the naive generalized estimating equation (GEE) method \citep{Liang1986Longitudinal}, the  naive EL method \citep{Qin1994Empirical}, and \cite{Lin2018Analysis}'s estimating equation method. When there are measurement errors in the covariate, the naive methods simply replace the unobserved true covariate with the average of replicate surrogate measurements.

\subsection{Simulation settings}
\label{ss:setting}
We consider the following linear regression model
\begin{equation*}
	Y_{ij}=\beta_{00}+X_{ij,1}\beta_{01}+X_{ij,2}\beta_{02}+\varepsilon_{ij}, \quad  i=1,\cdots,n, \quad j=1,\cdots,m,
\end{equation*}
where $(\beta_{00},\beta_{01},\beta_{02})=(1,1,1)$,  $m=6$, and the number of subjects $n$ is taken to be $50$, $100$, $200$, $300$, or $500$. The covariates $X_{ij,1}$ and $X_{ij,2}$ are drawn independently from the standard normal distribution $\mathcal{N}(0,1)$.  The random error $\bm{\varepsilon}_{i}=(\varepsilon_{i1},\cdots,\varepsilon_{im})^\top$  is generated from a multivariate normal distribution with mean zero and covariance matrix $\mathbf{R}_{i}({\rho}) {\sigma }_{e}^{2}$, where $\mathbf{R}_{i}({\rho})$ is the correlation matrix chosen to have an exchangeable  structure with $\rho=0.6$, and ${\sigma}_{e}^{2}=0.8$.

Assuming that there are measurement errors in the covariate $X_{ij,1}$, the surrogate values $W_{ij(k),1}$, $k=1,\cdots, K$ are generated from the following additive measurement error model:
\[
W_{ij(k),1}=X_{ij,1}+\xi_{ij(k)}, \quad  k=1,\cdots, K.
\]

To investigate the impacts of measurement errors on estimation  accuracy and efficiency, we consider the following  four cases, respectively.
\begin{enumerate}
	\item[C1:] There are two replicate measurements for $X_{ij,1}$. $\xi_{ij(1)}$ and $\xi_{ij(2)}$ are independently generated from a normal distribution with mean zero and standard deviation $0.6$.
	
	\item[C2:] There are two replicate measurements for $X_{ij,1}$. $\xi_{ij(1)}$ is generated from a normal distribution with mean zero and standard deviation $0.6$, while $\xi_{ij(2)}$ is generated from a $t$-distribution with $4$ degrees of freedom. $\xi_{ij(1)}$ and $\xi_{ij(2)}$ are independent.
	
	\item[C3:] There are three replicate measurements for $X_{ij,1}$. $\xi_{ij(1)}$, $\xi_{ij(2)}$, and $\xi_{ij(3)}$ are independently generated from a normal distribution with mean zero and standard deviation $0.6$.
	
	\item[C4:] There are three replicate measurements for $X_{ij,1}$. $\xi_{ij(1)}$ is generated from a normal distribution with mean zero and standard deviation $0.6$, $\xi_{ij(2)}$ is generated from a $t$-distribution with $4$ degrees of freedom, and $\xi_{ij(3)}$ is first generated from an exponential distribution  with the rate parameter $\lambda=2$ and then is centralized by subtracting its expectation $0.5$. $\xi_{ij(1)}$, $\xi_{ij(2)}$, and $\xi_{ij(3)}$ are independent.	
\end{enumerate}

Under each simulation setting, $1000$ replications are conducted.

\subsection{Simulation results}
\label{ss:result}
For each method, we calculate  the bias,  standard deviation (SD),  mean squared error (MSE), and coverage probability (CP) and mean length (ML) of the $95\%$ confidence interval.  For the GEE-based methods, the confidence intervals are constructed based on the asymptotic normality of the estimators. For the EL-based methods, the confidence intervals are constructed based on the asymptotic chi-squared distribution of the profile empirical likelihood ratio test statistic. The simulation results are presented in Tables~\ref{table:estcase1}--\ref{table:cicase4}.

By comparison, we find that the results of the  naive GEE method and the naive EL method are very similar. This is consistent with the statement that the GEE method and EL method are asymptotically  equivalent when the number of elements in the reduced  auxiliary random vector equals  the number of unknown parameters \citep{Qin1994Empirical}. When there are measurement errors in the covariate $X_{ij,1}$, the biases and MSEs of $\hat{\beta}_{1}$ obtained from the  two naive  methods  are all very large. Besides, the CPs for $\beta_{01}$ are very close to zero based on the  two naive  methods, which means that the true value always does not fall into the $95\%$ confidence interval.  Therefore, the effects of measurement errors cannot be ignored.  However, the biases of $\hat{\beta}_{1}$  based on \cite{Lin2018Analysis}'s method and the proposed method are much  smaller than those based on the  two naive  methods,  indicating that the biases induced by measurement errors can be eliminated successfully. When the distributions of different replicate measurements  are the same (C1 and C3), as the number of subjects increases,  \cite{Lin2018Analysis}'s method and the proposed method tend to be comparable with similar SD, MSE, and ML.  However, when the distributions of different replicate measurements are different (C2 and C4), the proposed method is more efficient than \cite{Lin2018Analysis}'s method when the number of subjects is not too small, because the SD and ML of the proposed method are generally smaller than those of \cite{Lin2018Analysis}'s method. 
%Note that as the number of subjects is taken to be $50$ or $100$, the SD of $\hat{\beta}_{0}$ based on \cite{Lin2018Analysis}'s method is smaller than that based on the proposed method. This is possible since the computational complexity of the proposed method is higher.

Comparing the results when the measurement errors are generated by way of C1 with those when the measurement errors are generated by way of C3, we can find that as the number of replicate measurements for $X_{ij,1}$ increases, all the methods become more efficient if the number of subjects is not too small. From the results when the measurement errors are generated by ways of C2 and C4,  we can get the same conclusion. Therefore, we need to make full use of information from all the replicate measurements. In addition,  the CPs of \cite{Lin2018Analysis}'s method and the proposed method are all close to the nominal confidence level $95\%$ if the number of subjects is larger than $50$, which shows that the confidence intervals obtained from the asymptotic theories are acceptable when the number of subjects is not too small.

It is worth mentioning that when  the number of subjects is small and the number of elements in the reduced auxiliary random vector is large, the empirical likelihood-based methods may have some problems. For example, \cite{han2014multiply} stated that their proposed empirical likelihood-based method might have numerical issues when the sample size (i.e., number of subjects) was small and/or the number of constraints  (i.e., number of elements in the reduced auxiliary random vector) was large. This may explain why the SD and MSE of the proposed method are sometimes larger than those of the \cite{Lin2018Analysis}'s method  when the number of subjects is $50$ and the measurement errors are generated by ways of C3 and C4, under which the number of elements in the reduced auxiliary random vector is $11$. \cite{tsao2004bounds} also showed that the least upper bounds on coverage probabilities of the empirical likelihood ratio confidence regions might be surprisingly small when the ratio of the number of subjects and the number of elements in the reduced auxiliary random vector was small. This may be the reason why the CPs of confidence intervals based on the proposed EL estimator are  lower than the nominal level $95\%$ when the number of subjects  is $50$.

 In summary, the proposed method performs well when the number of subjects is not too small. Specifically, it can eliminate the effects of  measurement errors in the covariate and has a high estimation efficiency.

 In order to obtain the simulation results in Tables \ref{table:estcase1}--\ref{table:cicase4}, we compute on the Digital Research Alliance of Canada's cluster Graham and use the R software (version 3.6.1).  The R codes  are available on the RunMyCode website. Table \ref{table:time} shows the average computation time for one replication. It can be found that the EL-based methods take more time than the GEE-based methods. This is not surprising because we need to perform more optimization calculations for the EL-based methods, especially when we calculate the confidence intervals. Thus, we will lose some computation efficiency to gain estimation efficiency. There is a trade-off between them.  If we place more emphasis on the estimation efficiency, then the proposed method is a good choice.

\begin{table}[!htp]
	% TABLE 1
	\centering
	\begin{threeparttable}
		\caption{Bias, standard deviation (SD) and mean squared error (MSE) of the estimator  when the measurement errors are generated by  way of C1} 		
		\label{table:estcase1}
		\setlength{\tabcolsep}{1pt}
		\begin{tabular*}{1.05\textwidth}{@{}r@{\extracolsep{\fill}}c@{\extracolsep{\fill}}c@{\extracolsep{\fill}}c@{\extracolsep{\fill}}c@{\extracolsep{\fill}}c@{\extracolsep{\fill}}c@{\extracolsep{\fill}}c@{\extracolsep{\fill}}c@{\extracolsep{\fill}}c@{\extracolsep{\fill}}c@{\extracolsep{\fill}}c@{\extracolsep{\fill}}c@{}}
			\toprule
			\multicolumn{1}{c}{$n$}&\multicolumn{4}{c}{${\beta}_{00}=1$}&\multicolumn{4}{c}{${\beta}_{01}=1$}&\multicolumn{4}{c}{${\beta}_{02}=1$}\\
			\cline{2-5} \cline{6-9}\cline{10-13}\\[-1.8em]
			& GEEN&ELN&LIN&Proposed&GEEN&ELN&LIN&Proposed&GEEN&ELN&LIN&Proposed   \\
			\midrule
  \multicolumn{1}{c}{Bias }	 & & & && && &&&&&\\				
50 & -0.49 & -0.49 & -0.46 & -0.34 & -14.93 & -14.93 & 0.74 & 0.69 & 0.01 & 0.01 & 0.04 & 0.05 \\ 
100 & 0.00 & 0.00 & 0.01 & 0.04 & -15.27 & -15.27 & 0.19 & 0.13 & 0.06 & 0.06 & 0.08 & 0.08 \\ 
200 & -0.03 & -0.03 & -0.03 & -0.03 & -15.15 & -15.15 & 0.24 & 0.22 & -0.11 & -0.11 & -0.12 & -0.11 \\ 
300 & 0.03 & 0.03 & 0.04 & 0.02 & -15.28 & -15.28 & 0.03 & 0.02 & -0.09 & -0.09 & -0.10 & -0.10 \\ 
500 & -0.05 & -0.05 & -0.05 & -0.05 & -15.22 & -15.22 & 0.11 & 0.11 & 0.02 & 0.02 & 0.04 & 0.04 \\ 

  \multicolumn{1}{c}{SD }	 & & & && && &&&&&\\				
50 & 10.55 & 10.55 & 10.61 & 10.93 & 4.10 & 4.10 & 5.33 & 5.50 & 4.34 & 4.34 & 4.43 & 4.63 \\ 
100 & 7.32 & 7.32 & 7.39 & 7.53 & 2.95 & 2.95 & 3.78 & 3.86 & 3.12 & 3.12 & 3.21 & 3.29 \\ 
200 & 5.18 & 5.18 & 5.19 & 5.21 & 2.00 & 2.00 & 2.64 & 2.67 & 2.15 & 2.15 & 2.18 & 2.21 \\ 
300 & 4.35 & 4.35 & 4.37 & 4.40 & 1.64 & 1.64 & 2.13 & 2.15 & 1.70 & 1.70 & 1.76 & 1.78 \\ 
500 & 3.34 & 3.34 & 3.36 & 3.37 & 1.27 & 1.27 & 1.67 & 1.67 & 1.34 & 1.34 & 1.38 & 1.38 \\ 
  
  \multicolumn{1}{c}{MSE }	 & & & && && &&&&&\\
50 & 1.12 & 1.12 & 1.13 & 1.19 & 2.40 & 2.40 & 0.29 & 0.31 & 0.19 & 0.19 & 0.20 & 0.21 \\ 
100 & 0.54 & 0.54 & 0.55 & 0.57 & 2.42 & 2.42 & 0.14 & 0.15 & 0.10 & 0.10 & 0.10 & 0.11 \\ 
200 & 0.27 & 0.27 & 0.27 & 0.27 & 2.33 & 2.33 & 0.07 & 0.07 & 0.05 & 0.05 & 0.05 & 0.05 \\ 
300 & 0.19 & 0.19 & 0.19 & 0.19 & 2.36 & 2.36 & 0.05 & 0.05 & 0.03 & 0.03 & 0.03 & 0.03 \\ 
500 & 0.11 & 0.11 & 0.11 & 0.11 & 2.33 & 2.33 & 0.03 & 0.03 & 0.02 & 0.02 & 0.02 & 0.02 \\

			\bottomrule
		\end{tabular*}
		\begin{tablenotes}	
			\item Note: \textbf{All the values of simulation results  are multiplied by $\bm{100}$}.  GEEN: the naive GEE method; ELN: the naive EL method; LIN: \cite{Lin2018Analysis}'s estimating equation method;  Proposed: the proposed EL method.
		\end{tablenotes}		
	\end{threeparttable}	
\end{table}

\begin{table}[!htp]
	% TABLE 2
	\centering
	\begin{threeparttable}
		\caption{Coverage probability (CP) and mean length (ML) of the $95\%$ confidence interval  when the measurement errors are generated by way of C1} 		
		\label{table:cicase1}
		\setlength{\tabcolsep}{1pt}
		\begin{tabular*}{1.05\textwidth}{@{}r@{\extracolsep{\fill}}c@{\extracolsep{\fill}}c@{\extracolsep{\fill}}c@{\extracolsep{\fill}}c@{\extracolsep{\fill}}c@{\extracolsep{\fill}}c@{\extracolsep{\fill}}c@{\extracolsep{\fill}}c@{\extracolsep{\fill}}c@{\extracolsep{\fill}}c@{\extracolsep{\fill}}c@{\extracolsep{\fill}}c@{}}
			\toprule
			\multicolumn{1}{c}{$n$}&\multicolumn{4}{c}{${\beta}_{00}=1$}&\multicolumn{4}{c}{${\beta}_{01}=1$}&\multicolumn{4}{c}{${\beta}_{02}=1$}\\
			\cline{2-5} \cline{6-9}\cline{10-13}\\[-1.8em]
				& GEEN&ELN&LIN&Proposed&GEEN&ELN&LIN&Proposed&GEEN&ELN&LIN&Proposed   \\
			\midrule
			\multicolumn{1}{c}{CP }	 & & & && && &&&&&\\				
50 & 93.2 & 93.5 & 93.1 & 92.8 & 3.9 & 4.4 & 93.5 & 91.5 & 94.0 & 94.1 & 94.2 & 92.2 \\ 
100 & 94.9 & 95.0 & 95.3 & 94.5 & 0.0 & 0.0 & 93.9 & 93.9 & 94.3 & 94.4 & 94.3 & 94.3 \\ 
200 & 95.3 & 95.5 & 95.1 & 95.0 & 0.0 & 0.0 & 93.6 & 93.8 & 94.5 & 94.4 & 94.4 & 94.1 \\ 
300 & 95.5 & 95.5 & 95.2 & 95.4 & 0.0 & 0.0 & 94.2 & 94.2 & 95.6 & 95.6 & 95.2 & 95.7 \\ 
500 & 95.1 & 95.1 & 95.2 & 95.1 & 0.0 & 0.0 & 93.4 & 93.1 & 95.2 & 95.3 & 94.7 & 94.4 \\

			\multicolumn{1}{c}{ML }	 & & & && && &&&&&\\				
50 & 40.8 & 41.5 & 41.0 & 40.5 & 15.2 & 15.4 & 20.0 & 19.9 & 16.5 & 16.8 & 17.0 & 16.9 \\ 
100 & 29.0 & 29.2 & 29.1 & 29.2 & 10.9 & 11.0 & 14.3 & 14.3 & 11.7 & 11.9 & 12.1 & 12.1 \\ 
200 & 20.5 & 20.6 & 20.6 & 20.7 & 7.7 & 7.7 & 10.1 & 10.1 & 8.4 & 8.4 & 8.6 & 8.6 \\ 
300 & 16.9 & 16.9 & 16.9 & 17.0 & 6.3 & 6.3 & 8.3 & 8.3 & 6.8 & 6.8 & 7.0 & 7.0 \\ 
500 & 13.1 & 13.1 & 13.1 & 13.1 & 4.9 & 4.9 & 6.4 & 6.4 & 5.3 & 5.3 & 5.5 & 5.5 \\ 
					
			\bottomrule
		\end{tabular*}
		\begin{tablenotes}	
			\item Note: \textbf{All the values of simulation results  are multiplied by $\bm{100}$}.   GEEN: the naive GEE method; ELN: the naive EL method; LIN: \cite{Lin2018Analysis}'s estimating equation method;  Proposed: the proposed EL method.
		\end{tablenotes}		
	\end{threeparttable}	
\end{table}

\begin{table}[!htp]
	% TABLE 3
	\centering
	\begin{threeparttable}
		\caption{Bias, standard deviation (SD) and mean squared error (MSE) of the estimator when the measurement errors are generated by  way of C2} 		
		\label{table:estcase2}
		\setlength{\tabcolsep}{1pt}
		\begin{tabular*}{1.05\textwidth}{@{}r@{\extracolsep{\fill}}c@{\extracolsep{\fill}}c@{\extracolsep{\fill}}c@{\extracolsep{\fill}}c@{\extracolsep{\fill}}c@{\extracolsep{\fill}}c@{\extracolsep{\fill}}c@{\extracolsep{\fill}}c@{\extracolsep{\fill}}c@{\extracolsep{\fill}}c@{\extracolsep{\fill}}c@{\extracolsep{\fill}}c@{}}
			\toprule
			\multicolumn{1}{c}{$n$}&\multicolumn{4}{c}{${\beta}_{00}=1$}&\multicolumn{4}{c}{${\beta}_{01}=1$}&\multicolumn{4}{c}{${\beta}_{02}=1$}\\
			\cline{2-5} \cline{6-9}\cline{10-13}\\[-1.8em]
			& GEEN&ELN&LIN&Proposed&GEEN&ELN&LIN&Proposed&GEEN&ELN&LIN&Proposed   \\
			\midrule
			\multicolumn{1}{c}{Bias }	 & & & && && &&&&&\\			
50 & -0.57 & -0.57 & -0.52 & -0.41 & -36.72 & -36.72 & 1.32 & 0.78 & -0.07 & -0.08 & -0.07 & 0.12 \\ 
100 & -0.04 & -0.05 & -0.06 & -0.03 & -37.06 & -37.08 & 0.65 & 0.47 & 0.01 & 0.01 & 0.05 & 0.03 \\ 
200 & 0.02 & 0.02 & 0.06 & -0.01 & -36.90 & -36.90 & 0.50 & 0.43 & -0.06 & -0.06 & -0.06 & -0.08 \\ 
300 & 0.01 & 0.01 & 0.01 & 0.00 & -36.91 & -36.92 & 0.31 & 0.30 & -0.03 & -0.03 & -0.04 & -0.04 \\ 
500 & -0.01 & -0.01 & 0.01 & -0.03 & -37.08 & -37.11 & 0.13 & 0.12 & -0.01 & -0.00 & 0.04 & 0.02 \\ 
			
			\multicolumn{1}{c}{SD }	 & & & && && &&&&&\\
50 & 10.88 & 10.88 & 11.28 & 11.27 & 5.33 & 5.35 & 8.16 & 8.36 & 5.49 & 5.50 & 6.31 & 5.46 \\ 
100 & 7.50 & 7.51 & 7.83 & 7.70 & 3.86 & 3.90 & 5.85 & 5.73 & 3.69 & 3.70 & 4.17 & 3.54 \\ 
200 & 5.32 & 5.32 & 5.47 & 5.35 & 2.80 & 2.81 & 4.32 & 4.29 & 2.65 & 2.66 & 3.01 & 2.49 \\ 
300 & 4.47 & 4.47 & 4.62 & 4.47 & 2.24 & 2.27 & 3.35 & 3.30 & 1.98 & 1.98 & 2.32 & 1.96 \\ 
500 & 3.43 & 3.43 & 3.55 & 3.43 & 1.88 & 1.96 & 2.59 & 2.57 & 1.68 & 1.68 & 1.92 & 1.59 \\ 
			
			\multicolumn{1}{c}{MSE }	 & & & && && &&&&&\\
50 & 1.18 & 1.18 & 1.27 & 1.27 & 13.76 & 13.77 & 0.68 & 0.70 & 0.30 & 0.30 & 0.40 & 0.30 \\ 
100 & 0.56 & 0.56 & 0.61 & 0.59 & 13.88 & 13.90 & 0.35 & 0.33 & 0.14 & 0.14 & 0.17 & 0.13 \\ 
200 & 0.28 & 0.28 & 0.30 & 0.29 & 13.69 & 13.70 & 0.19 & 0.19 & 0.07 & 0.07 & 0.09 & 0.06 \\ 
300 & 0.20 & 0.20 & 0.21 & 0.20 & 13.67 & 13.68 & 0.11 & 0.11 & 0.04 & 0.04 & 0.05 & 0.04 \\ 
500 & 0.12 & 0.12 & 0.13 & 0.12 & 13.79 & 13.81 & 0.07 & 0.07 & 0.03 & 0.03 & 0.04 & 0.03 \\
			\bottomrule
		\end{tabular*}
		\begin{tablenotes}	
			\item Note: \textbf{All the values of simulation results  are multiplied by $\bm{100}$}.  GEEN: the naive GEE method; ELN: the naive EL method; LIN: \cite{Lin2018Analysis}'s estimating equation method;  Proposed: the proposed EL method.
		\end{tablenotes}		
	\end{threeparttable}	
\end{table}

\begin{table}[!htp]
	% TABLE 4
	\centering
	\begin{threeparttable}
		\caption{Coverage probability (CP) and mean length (ML) of the $95\%$ confidence interval  when the measurement errors are generated by  way of C2} 		
		\label{table:cicase2}
		\setlength{\tabcolsep}{1pt}
		\begin{tabular*}{1.05\textwidth}{@{}r@{\extracolsep{\fill}}c@{\extracolsep{\fill}}c@{\extracolsep{\fill}}c@{\extracolsep{\fill}}c@{\extracolsep{\fill}}c@{\extracolsep{\fill}}c@{\extracolsep{\fill}}c@{\extracolsep{\fill}}c@{\extracolsep{\fill}}c@{\extracolsep{\fill}}c@{\extracolsep{\fill}}c@{\extracolsep{\fill}}c@{}}
			\toprule
			\multicolumn{1}{c}{$n$}&\multicolumn{4}{c}{${\beta}_{00}=1$}&\multicolumn{4}{c}{${\beta}_{01}=1$}&\multicolumn{4}{c}{${\beta}_{02}=1$}\\
			\cline{2-5} \cline{6-9}\cline{10-13}\\[-1.8em]
			& GEEN&ELN&LIN&Proposed&GEEN&ELN&LIN&Proposed&GEEN&ELN&LIN&Proposed   \\
			\midrule
			\multicolumn{1}{c}{CP }	 & & & && && &&&&&\\			
50 & 93.5 & 93.7 & 93.2 & 92.4 & 0.0 & 0.0 & 95.1 & 92.3 & 93.5 & 94.0 & 93.8 & 91.1 \\ 
100 & 95.0 & 95.1 & 95.2 & 94.6 & 0.0 & 0.0 & 95.2 & 94.3 & 95.1 & 95.2 & 94.2 & 93.9 \\ 
200 & 94.6 & 94.6 & 94.9 & 94.7 & 0.0 & 0.0 & 95.2 & 93.8 & 94.2 & 94.3 & 95.4 & 95.1 \\ 
300 & 95.2 & 95.0 & 95.5 & 95.4 & 0.0 & 0.0 & 95.3 & 94.2 & 95.7 & 95.8 & 95.6 & 96.1 \\ 
500 & 95.6 & 95.5 & 94.9 & 95.4 & 0.0 & 0.0 & 95.2 & 95.0 & 94.8 & 94.7 & 94.9 & 94.5 \\

			\multicolumn{1}{c}{ML }	 & & & && && &&&&&\\				
50 & 42.1 & 42.8 & 43.6 & 41.5 & 19.1 & 19.3 & 33.0 & 31.9 & 19.9 & 20.2 & 23.1 & 19.2 \\ 
100 & 29.9 & 30.2 & 30.9 & 29.8 & 14.0 & 14.2 & 23.2 & 22.8 & 14.1 & 14.2 & 16.3 & 13.7 \\ 
200 & 21.2 & 21.3 & 21.9 & 21.1 & 10.2 & 10.5 & 16.4 & 16.1 & 10.0 & 10.1 & 11.5 & 9.7 \\ 
300 & 17.4 & 17.4 & 17.9 & 17.3 & 8.4 & 8.6 & 13.3 & 13.1 & 8.2 & 8.3 & 9.4 & 7.9 \\ 
500 & 13.5 & 13.5 & 13.9 & 13.4 & 6.8 & 6.9 & 10.3 & 10.1 & 6.4 & 6.4 & 7.3 & 6.1 \\ 
			
			\bottomrule
		\end{tabular*}
		\begin{tablenotes}	
			\item Note:    \textbf{All the values of simulation results  are multiplied by $\bm{100}$}.   GEEN: the naive GEE method; ELN: the naive EL method; LIN: \cite{Lin2018Analysis}'s estimating equation method;  Proposed: the proposed EL method.
		\end{tablenotes}		
	\end{threeparttable}	
\end{table}

\begin{table}[!htp]
	% TABLE 5
	\centering
	\begin{threeparttable}
		\caption{Bias, standard deviation (SD) and mean squared error (MSE) of the estimator when the measurement errors are generated by  way of C3} 		
		\label{table:estcase3}
		\setlength{\tabcolsep}{1pt}
		\begin{tabular*}{1.05\textwidth}{@{}r@{\extracolsep{\fill}}c@{\extracolsep{\fill}}c@{\extracolsep{\fill}}c@{\extracolsep{\fill}}c@{\extracolsep{\fill}}c@{\extracolsep{\fill}}c@{\extracolsep{\fill}}c@{\extracolsep{\fill}}c@{\extracolsep{\fill}}c@{\extracolsep{\fill}}c@{\extracolsep{\fill}}c@{\extracolsep{\fill}}c@{}}
			\toprule
			\multicolumn{1}{c}{$n$}&\multicolumn{4}{c}{${\beta}_{00}=1$}&\multicolumn{4}{c}{${\beta}_{01}=1$}&\multicolumn{4}{c}{${\beta}_{02}=1$}\\
			\cline{2-5} \cline{6-9}\cline{10-13}\\[-1.8em]
				& GEEN&ELN&LIN&Proposed&GEEN&ELN&LIN&Proposed&GEEN&ELN&LIN&Proposed   \\
			\midrule
			\multicolumn{1}{c}{Bias }	 & & & && && &&&&&\\			
50 & -0.46 & -0.46 & -0.43 & -0.22 & -10.59 & -10.59 & 0.35 & 0.20 & 0.03 & 0.03 & 0.05 & 0.10 \\ 
100 & 0.04 & 0.04 & 0.04 & 0.15 & -10.72 & -10.72 & 0.12 & -0.01 & 0.06 & 0.06 & 0.06 & 0.08 \\ 
200 & -0.01 & -0.01 & -0.01 & 0.02 & -10.75 & -10.75 & 0.02 & -0.01 & -0.06 & -0.06 & -0.07 & -0.04 \\ 
300 & 0.03 & 0.03 & 0.03 & 0.02 & -10.74 & -10.74 & 0.01 & 0.00 & -0.07 & -0.07 & -0.07 & -0.06 \\ 
500 & -0.05 & -0.05 & -0.05 & -0.06 & -10.70 & -10.70 & 0.04 & 0.05 & 0.02 & 0.02 & 0.02 & 0.02 \\ 
			
			\multicolumn{1}{c}{SD }	 & & & && && &&&&&\\
50 & 10.50 & 10.50 & 10.53 & 11.78 & 4.01 & 4.01 & 4.67 & 5.21 & 4.21 & 4.21 & 4.24 & 4.91 \\ 
100 & 7.31 & 7.31 & 7.35 & 7.80 & 2.82 & 2.82 & 3.25 & 3.47 & 2.89 & 2.89 & 2.95 & 3.15 \\ 
200 & 5.16 & 5.16 & 5.17 & 5.29 & 2.01 & 2.01 & 2.38 & 2.45 & 2.04 & 2.04 & 2.06 & 2.11 \\ 
300 & 4.31 & 4.31 & 4.33 & 4.39 & 1.56 & 1.56 & 1.83 & 1.87 & 1.63 & 1.63 & 1.66 & 1.68 \\ 
500 & 3.32 & 3.32 & 3.33 & 3.37 & 1.24 & 1.24 & 1.45 & 1.46 & 1.27 & 1.27 & 1.30 & 1.30 \\ 
			
			\multicolumn{1}{c}{MSE }	 & & & && && &&&&&\\
50 & 1.10 & 1.10 & 1.11 & 1.39 & 1.28 & 1.28 & 0.22 & 0.27 & 0.18 & 0.18 & 0.18 & 0.24 \\ 
100 & 0.53 & 0.53 & 0.54 & 0.61 & 1.23 & 1.23 & 0.11 & 0.12 & 0.08 & 0.08 & 0.09 & 0.10 \\ 
200 & 0.27 & 0.27 & 0.27 & 0.28 & 1.20 & 1.20 & 0.06 & 0.06 & 0.04 & 0.04 & 0.04 & 0.04 \\ 
300 & 0.19 & 0.19 & 0.19 & 0.19 & 1.18 & 1.18 & 0.03 & 0.03 & 0.03 & 0.03 & 0.03 & 0.03 \\ 
500 & 0.11 & 0.11 & 0.11 & 0.11 & 1.16 & 1.16 & 0.02 & 0.02 & 0.02 & 0.02 & 0.02 & 0.02 \\
			\bottomrule
		\end{tabular*}
		\begin{tablenotes}	
			\item Note: \textbf{All the values of simulation results  are multiplied by $\bm{100}$}.  GEEN: the naive GEE method; ELN: the naive EL method; LIN: \cite{Lin2018Analysis}'s estimating equation method;  Proposed: the proposed EL method.
		\end{tablenotes}		
	\end{threeparttable}	
\end{table}

\begin{table}[!htp]
	% TABLE 6
	\centering
	\begin{threeparttable}
		\caption{Coverage probability (CP) and mean length (ML) of the $95\%$ confidence interval when the measurement errors are generated by  way of C3} 		
		\label{table:cicase3}
		\setlength{\tabcolsep}{1pt}
		\begin{tabular*}{1.05\textwidth}{@{}r@{\extracolsep{\fill}}c@{\extracolsep{\fill}}c@{\extracolsep{\fill}}c@{\extracolsep{\fill}}c@{\extracolsep{\fill}}c@{\extracolsep{\fill}}c@{\extracolsep{\fill}}c@{\extracolsep{\fill}}c@{\extracolsep{\fill}}c@{\extracolsep{\fill}}c@{\extracolsep{\fill}}c@{\extracolsep{\fill}}c@{}}
			\toprule
			\multicolumn{1}{c}{$n$}&\multicolumn{4}{c}{${\beta}_{00}=1$}&\multicolumn{4}{c}{${\beta}_{01}=1$}&\multicolumn{4}{c}{${\beta}_{02}=1$}\\
			\cline{2-5} \cline{6-9}\cline{10-13}\\[-1.8em]
			& GEEN&ELN&LIN&Proposed&GEEN&ELN&LIN&Proposed&GEEN&ELN&LIN&Proposed   \\
			\midrule
			\multicolumn{1}{c}{CP }	 & & & && && &&&&&\\			
50 & 93.8 & 93.8 & 93.5 & 87.4 & 21.9 & 22.9 & 93.4 & 86.3 & 92.8 & 93.1 & 92.8 & 84.8 \\ 
100 & 94.7 & 94.9 & 94.4 & 93.1 & 2.5 & 2.9 & 94.0 & 91.8 & 94.4 & 94.6 & 94.4 & 91.9 \\ 
200 & 95.3 & 95.3 & 95.1 & 94.9 & 0.0 & 0.0 & 94.0 & 93.4 & 94.7 & 94.6 & 94.3 & 93.5 \\ 
300 & 95.1 & 95.2 & 95.2 & 94.9 & 0.0 & 0.0 & 95.2 & 94.6 & 94.9 & 95.0 & 94.8 & 94.2 \\ 
500 & 94.7 & 94.6 & 94.6 & 94.6 & 0.0 & 0.0 & 94.3 & 94.3 & 94.9 & 94.9 & 94.6 & 95.0 \\

			\multicolumn{1}{c}{ML }	 & & & && && &&&&&\\			
50 & 40.5 & 41.2 & 40.6 & 36.8 & 14.8 & 15.1 & 17.4 & 15.7 & 15.8 & 16.0 & 16.0 & 14.5 \\ 
100 & 28.8 & 29.0 & 28.8 & 28.3 & 10.6 & 10.7 & 12.4 & 12.2 & 11.2 & 11.3 & 11.4 & 11.1 \\ 
200 & 20.4 & 20.5 & 20.4 & 20.4 & 7.5 & 7.5 & 8.8 & 8.8 & 7.9 & 8.0 & 8.1 & 8.0 \\ 
300 & 16.8 & 16.8 & 16.8 & 16.8 & 6.2 & 6.2 & 7.2 & 7.2 & 6.5 & 6.5 & 6.6 & 6.6 \\ 
500 & 13.0 & 13.0 & 13.0 & 13.0 & 4.8 & 4.8 & 5.6 & 5.6 & 5.0 & 5.1 & 5.1 & 5.1 \\ 
			\bottomrule
		\end{tabular*}
		\begin{tablenotes}	
			\item Note: \textbf{All the values of simulation results  are multiplied by $\bm{100}$}.   GEEN: the naive GEE method; ELN: the naive EL method; LIN: \cite{Lin2018Analysis}'s estimating equation method;  Proposed: the proposed EL method.
		\end{tablenotes}		
	\end{threeparttable}	
\end{table}

\begin{table}[!htp]
	% TABLE 7
	\centering
	\begin{threeparttable}
		\caption{Bias, standard deviation (SD) and mean squared error (MSE) of the estimator  when the measurement errors are generated by  way of C4} 		
		\label{table:estcase4}
		\begin{tabular*}{1.05\textwidth}{@{}r@{\extracolsep{\fill}}c@{\extracolsep{\fill}}c@{\extracolsep{\fill}}c@{\extracolsep{\fill}}c@{\extracolsep{\fill}}c@{\extracolsep{\fill}}c@{\extracolsep{\fill}}c@{\extracolsep{\fill}}c@{\extracolsep{\fill}}c@{\extracolsep{\fill}}c@{\extracolsep{\fill}}c@{\extracolsep{\fill}}c@{}}
			\toprule
			\multicolumn{1}{c}{$n$}&\multicolumn{4}{c}{${\beta}_{00}=1$}&\multicolumn{4}{c}{${\beta}_{01}=1$}&\multicolumn{4}{c}{${\beta}_{02}=1$}\\
			\cline{2-5} \cline{6-9}\cline{10-13}\\[-1.8em]
			& GEEN&ELN&LIN&Proposed&GEEN&ELN&LIN&Proposed&GEEN&ELN&LIN&Proposed   \\
			\midrule
			\multicolumn{1}{c}{Bias }	 & & & && && &&&&&\\			
50 & -0.51 & -0.51 & -0.47 & -0.29 & -22.09 & -22.09 & 0.81 & 0.41 & 0.05 & 0.05 & 0.05 & 0.15 \\ 
100 & -0.01 & -0.02 & -0.01 & 0.12 & -22.52 & -22.52 & 0.22 & 0.02 & 0.16 & 0.16 & 0.19 & 0.19 \\ 
200 & 0.03 & 0.03 & 0.05 & 0.01 & -22.43 & -22.43 & 0.13 & 0.08 & -0.09 & -0.09 & -0.09 & -0.12 \\ 
300 & 0.01 & 0.01 & 0.01 & -0.03 & -22.37 & -22.38 & 0.11 & 0.04 & -0.02 & -0.02 & -0.03 & -0.00 \\ 
500 & -0.02 & -0.03 & -0.01 & -0.05 & -22.45 & -22.47 & 0.09 & 0.04 & 0.05 & 0.05 & 0.08 & 0.08 \\
			
			\multicolumn{1}{c}{SD }	 & & & && && &&&&&\\
50 & 10.65 & 10.65 & 10.79 & 12.00 & 4.63 & 4.63 & 5.47 & 5.30 & 4.69 & 4.68 & 5.02 & 4.92 \\ 
100 & 7.35 & 7.35 & 7.47 & 7.61 & 3.34 & 3.34 & 4.02 & 3.61 & 3.14 & 3.14 & 3.35 & 3.15 \\ 
200 & 5.23 & 5.23 & 5.28 & 5.34 & 2.46 & 2.47 & 2.98 & 2.57 & 2.32 & 2.32 & 2.42 & 2.21 \\ 
300 & 4.36 & 4.36 & 4.40 & 4.35 & 1.99 & 1.99 & 2.37 & 2.01 & 1.86 & 1.86 & 2.00 & 1.79 \\ 
500 & 3.36 & 3.36 & 3.40 & 3.36 & 1.56 & 1.64 & 1.80 & 1.52 & 1.47 & 1.48 & 1.56 & 1.36 \\ 
			
			\multicolumn{1}{c}{MSE }	 & & & && && &&&&&\\
50 & 1.14 & 1.14 & 1.16 & 1.44 & 5.09 & 5.09 & 0.31 & 0.28 & 0.22 & 0.22 & 0.25 & 0.24 \\ 
100 & 0.54 & 0.54 & 0.56 & 0.58 & 5.18 & 5.18 & 0.16 & 0.13 & 0.10 & 0.10 & 0.11 & 0.10 \\ 
200 & 0.27 & 0.27 & 0.28 & 0.29 & 5.09 & 5.09 & 0.09 & 0.07 & 0.05 & 0.05 & 0.06 & 0.05 \\ 
300 & 0.19 & 0.19 & 0.19 & 0.19 & 5.04 & 5.05 & 0.06 & 0.04 & 0.03 & 0.03 & 0.04 & 0.03 \\ 
500 & 0.11 & 0.11 & 0.12 & 0.11 & 5.06 & 5.08 & 0.03 & 0.02 & 0.02 & 0.02 & 0.02 & 0.02 \\ 
			\bottomrule
		\end{tabular*}
		\begin{tablenotes}	
			\item Note: \textbf{All the values of simulation results  are multiplied by $\bm{100}$}.  GEEN: the naive GEE method; ELN: the naive EL method; LIN: \cite{Lin2018Analysis}'s estimating equation method;  Proposed: the proposed EL method.
		\end{tablenotes}		
	\end{threeparttable}	
\end{table}

\begin{table}[!htp]
	% TABLE 8
	\centering
	\begin{threeparttable}
		\caption{Coverage probability (CP) and mean length (ML) of the $95\%$ confidence interval when the measurement errors are generated by way of C4} 		
		\label{table:cicase4}
		\setlength{\tabcolsep}{1pt}
		\begin{tabular*}{1.05\textwidth}{@{}r@{\extracolsep{\fill}}c@{\extracolsep{\fill}}c@{\extracolsep{\fill}}c@{\extracolsep{\fill}}c@{\extracolsep{\fill}}c@{\extracolsep{\fill}}c@{\extracolsep{\fill}}c@{\extracolsep{\fill}}c@{\extracolsep{\fill}}c@{\extracolsep{\fill}}c@{\extracolsep{\fill}}c@{\extracolsep{\fill}}c@{}}
			\toprule
			\multicolumn{1}{c}{$n$}&\multicolumn{4}{c}{${\beta}_{00}=1$}&\multicolumn{4}{c}{${\beta}_{01}=1$}&\multicolumn{4}{c}{${\beta}_{02}=1$}\\
			\cline{2-5} \cline{6-9}\cline{10-13}\\[-1.8em]
			& GEEN&ELN&LIN&Proposed&GEEN&ELN&LIN&Proposed&GEEN&ELN&LIN&Proposed   \\
			\midrule
			\multicolumn{1}{c}{CP }	 & & & && && &&&&&\\				
50 & 93.5 & 93.5 & 93.0 & 87.9 & 0.2 & 0.2 & 95.6 & 86.9 & 94.6 & 95.0 & 93.9 & 84.9 \\ 
100 & 94.9 & 94.9 & 95.1 & 93.2 & 0.0 & 0.0 & 95.8 & 91.5 & 95.3 & 95.2 & 94.9 & 91.7 \\ 
200 & 94.8 & 94.5 & 94.3 & 93.6 & 0.0 & 0.0 & 94.2 & 92.6 & 94.5 & 94.6 & 95.0 & 94.3 \\ 
300 & 95.3 & 95.3 & 95.3 & 94.9 & 0.0 & 0.0 & 93.9 & 93.9 & 94.5 & 94.7 & 94.9 & 94.3 \\ 
500 & 94.9 & 94.7 & 94.9 & 94.4 & 0.0 & 0.0 & 95.0 & 94.9 & 94.4 & 94.3 & 93.9 & 94.9 \\

			\multicolumn{1}{c}{ML }	 & & & && && &&&&&\\					
50 & 41.3 & 41.9 & 41.7 & 37.2 & 17.4 & 17.7 & 22.3 & 16.6 & 17.7 & 18.0 & 18.9 & 14.8 \\ 
100 & 29.3 & 29.6 & 29.6 & 28.3 & 12.6 & 12.8 & 15.8 & 12.8 & 12.5 & 12.7 & 13.4 & 11.4 \\ 
200 & 20.8 & 20.9 & 21.0 & 20.4 & 9.1 & 9.3 & 11.2 & 9.2 & 8.9 & 9.0 & 9.5 & 8.2 \\ 
300 & 17.0 & 17.1 & 17.2 & 16.8 & 7.5 & 7.6 & 9.1 & 7.6 & 7.3 & 7.3 & 7.7 & 6.8 \\ 
500 & 13.2 & 13.2 & 13.3 & 13.0 & 5.9 & 6.0 & 7.1 & 5.9 & 5.7 & 5.7 & 6.0 & 5.3 \\
			\bottomrule
		\end{tabular*}
		\begin{tablenotes}	
			\item Note:  \textbf{All the values of simulation results  are multiplied by $\bm{100}$}. GEEN: the naive GEE method; ELN: the naive EL method; LIN: \cite{Lin2018Analysis}'s estimating equation method;  Proposed: the proposed EL method.
		\end{tablenotes}		
	\end{threeparttable}	
\end{table}

\begin{table}[!htp]
		% TABLE 9
	\centering
	\begin{threeparttable}
		\caption{Average computation time (seconds) for one replication}
		\label{table:time}
		\begin{tabular*}{\textwidth}{@{}l@{\extracolsep{\fill}}c@{\extracolsep{\fill}}c@{\extracolsep{\fill}}c@{\extracolsep{\fill}}c@{\extracolsep{\fill}}c@{\extracolsep{\fill}}c@{}}
			\toprule
		Way	&Method&50&100&200&300&500\\
			\midrule
			\multirow{4}{*}{C1}
&GEEN & 0.02 & 0.05 & 0.08 & 0.12 & 0.21 \\ 
&ELN & 96.89 & 132.54 & 186.17 & 234.11 & 379.63 \\ 
&LIN & 0.03 & 0.06 & 0.11 & 0.15 & 0.26 \\ 
&Proposed & 134.72 & 169.73 & 228.88 & 302.45 & 475.25 \\ 
\hline
\multirow{4}{*}{C2}
&GEEN & 0.02 & 0.03 & 0.08 & 0.12 & 0.20 \\ 
&ELN & 83.43 & 120.92 & 187.40 & 255.87 & 387.77 \\ 
&LIN & 0.02 & 0.05 & 0.10 & 0.15 & 0.26 \\ 
&Proposed & 119.65 & 147.34 & 219.67 & 300.61 & 452.90 \\ 
\hline
\multirow{4}{*}{C3}
&GEEN & 0.02 & 0.04 & 0.08 & 0.13 & 0.21 \\ 
&ELN & 88.98 & 116.06 & 167.73 & 213.84 & 336.93 \\ 
&LIN & 0.07 & 0.13 & 0.26 & 0.41 & 0.69 \\ 
&Proposed & 254.58 & 327.10 & 470.43 & 668.72 & 1097.78 \\
\hline
\multirow{4}{*}{C4}
&GEEN & 0.02 & 0.04 & 0.08 & 0.12 & 0.21 \\ 
&ELN & 91.90 & 114.59 & 166.00 & 224.82 & 383.31 \\ 
&LIN & 0.07 & 0.13 & 0.25 & 0.38 & 0.68 \\ 
&Proposed & 263.33 & 311.48 & 457.86 & 634.20 & 1094.20 \\ 
			\bottomrule
		\end{tabular*}
		\begin{tablenotes}	
			\item Note:  GEEN: the naive GEE method; ELN: the naive EL method; LIN: \cite{Lin2018Analysis}'s estimating equation method;  Proposed: the proposed EL method.
		\end{tablenotes}		
	\end{threeparttable}	
\end{table}

\section{Application to the LEAN study}
\label{s:realdata}
We apply the proposed method to the Lifestyle Education for Activity and Nutrition (LEAN) study \citep{Barry2011Using}. The intervention strategy was providing the participants with   a group-based behavioural weight loss program or/and  the SenseWear platform which could help improve lifestyle self-monitoring \citep{Shuger2011Electronic}. The data of age, gender, race, and education level were collected at baseline. Body weight, height, systolic blood pressure (SBP), and diastolic blood pressure (DBP) were measured at baseline, month $4$ and month $9$. Among the $197$ participants, $74$  participants failed to complete the month $4$ or/and month $9$ assessments.  For convenience, we exclude them from the current study.

The main interest of this study is to assess whether the intervention strategy is effective in reducing the BMI value at months $4$ and $9$. The response variable BMI is  calculated by $\mathrm{log(weight/height^{2}\times 703)}$. The main exposure variable is intervention, which is denoted as ``group''. Besides, it takes the value of $1$ for the intervention group and takes the value of $0$ for the standard care group. Two dummy variables $t_{1}$ and $t_{2}$ are introduced to represent the assessment time. That is, $t_{1}=t_{2}=0$ for baseline, $t_{1}=1$ and $t_{2}=0$ for month $4$, and $t_{1}=0$ and $t_{2}=1$ for month $9$. To reveal the effects of different groups at different  time points, we  include the interaction terms between group and time as covariates. Other covariates considered are SBP, DBP, age, gender (female, 1; male, 0), race (African American, 1; others, 0), and education level (four-year college or higher, 1; others, 0). To make the covariates have similar scales,  the values of SBP, DBP, and  age are divided by $100$. Before analysis, all the variables are centralized by subtracting their mean values. The following linear regression model is adopted to fit the LEAN data set:
\begin{equation}
	\begin{aligned}
		Y=&\mathrm{SBP} \, \beta_{1}+\mathrm{DBP} \, \beta_{2}+\mathrm{AGE} \, \beta_{3}+\mathrm{gender} \, \beta_{4}+\mathrm{race} \, \beta_{5}+\mathrm{education} \, \beta_{6}\\
		&+\mathrm{group} \, \beta_{7}+t_{1}\beta_{8}+t_{2}\beta_{9}+(\mathrm{group}  \times t_{1})\,\beta_{10}+(\mathrm{group} \times t_{2})\,\beta_{11}+\varepsilon.
	\end{aligned}
	\label{eq:rdlrm}
\end{equation}

As mentioned in Section~\ref{s:intro}, there exist measurement errors in the covariates SBP and DBP, and the replicate measurement errors of SBP follow different distributions, as does DBP. Therefore, it is necessary to consider the distributional difference of measurement errors to obtain an  efficient estimator. Table~\ref{table:3} displays the estimates of regression coefficients and confidence intervals by using different methods, including  the naive GEE method  \citep{Liang1986Longitudinal}, the naive EL method \citep{Qin1994Empirical},  \cite{Lin2018Analysis}'s estimating equation method, and the proposed EL method.  It can be found that the effect of intervention at month $9$ is significantly different from zero at the significance level of $0.05$ based on the naive GEE method, \cite{Lin2018Analysis}'s estimating equation method, and the proposed EL method. Besides, all these three methods show that intervention at month $9$ is negatively related to BMI, which means that the intervention  strategy can significantly reduce BMI at month $9$. This conclusion is consistent with the finding of LEAN's study group \citep{Shuger2011Electronic}. In addition, all the methods show that SBP is significantly positively related to BMI. Besides,  we find that the result of the naive GEE method is similar to that of  the  naive EL method. But the differences between these two naive methods and the proposed EL method are relatively large, which may be due to the effects of measurement errors. There are  some differences in the estimates of  \cite{Lin2018Analysis}'s estimating equation method and the proposed EL method. For example, the effect of  $t_{2}$ is significantly different from zero at the significance level of $0.05$ based on the proposed EL method, while it is not  based on \cite{Lin2018Analysis}'s estimating equation method. This is perhaps because the confidence intervals for these two methods are constructed in different ways and the confidence intervals of the proposed method are shorter. In general, the proposed method has the shortest average length of  confidence intervals. Therefore, the result of the proposed method is  recommended.

\begin{sidewaystable}[!htp]
	% TABLE 10
	\centering
	\linespread{1.3}
	\begin{threeparttable}
		\caption{Estimates of regression coefficients and the 95\% confidence intervals in the analysis of the LEAN data set}
		\label{table:3}
		\begin{tabular*}{1.05\textwidth}{@{}l@{\extracolsep{\fill}}c@{\extracolsep{\fill}}c@{\extracolsep{\fill}}c@{\extracolsep{\fill}}c@{\extracolsep{\fill}}c@{\extracolsep{\fill}}c@{\extracolsep{\fill}}c@{\extracolsep{\fill}}c@{\extracolsep{\fill}}c@{\extracolsep{\fill}}c@{\extracolsep{\fill}}c@{\extracolsep{\fill}}c@{\extracolsep{\fill}}c@{\extracolsep{\fill}}c@{\extracolsep{\fill}}c@{\extracolsep{\fill}}c@{}}
			\hline
			\multicolumn{1}{c}{}&\multicolumn{4}{c}{GEEN}&\multicolumn{4}{c}{ELN}&\multicolumn{4}{c}{LIN}&\multicolumn{4}{c}{Proposed}\\[1ex]
			&Coef&Lower&Upper& CL&Coef&Lower&Upper& CL&Coef&Lower&Upper& CL&Coef&Lower&Upper& CL\\\hline
			
			SBP	&	$10.22^{\ast}$ 	&	2.90 	&	17.53 	&	14.63 	&	$9.57^{\ast}$ 	&	2.34 	&	17.30 	&	14.96 	&	$8.43^{\ast}$ 	&	0.78 	&	16.08 	&	15.30 	&	$9.16^{\ast}$ 	&	3.20 	&	11.29 	&	8.09 	\\
			DBP	&	5.92 	&	-4.86 	&	16.70 	&	21.56 	&	5.32 	&	-5.48 	&	16.19 	&	21.67 	&	4.09 	&	-7.56 	&	15.73 	&	23.29 	&	2.47 	&	-4.43 	&	4.51 	&	8.94 	\\
			AGE	&	-2.47 	&	-26.65 	&	21.71 	&	48.36 	&	-2.32 	&	-26.40 	&	22.18 	&	48.58 	&	-2.08 	&	-26.43 	&	22.27 	&	48.71 	&	-5.67 	&	-29.29 	&	18.05 	&	47.34 	\\
			gender	&	0.68 	&	-5.55 	&	6.91 	&	12.46 	&	0.62 	&	-5.93 	&	6.69 	&	12.62 	&	0.53 	&	-5.71 	&	6.76 	&	12.47 	&	2.08 	&	-3.52 	&	7.55 	&	11.07 	\\
			race	&	3.90 	&	-2.05 	&	9.85 	&	11.90 	&	3.97 	&	-1.98 	&	9.95 	&	11.93 	&	4.11 	&	-1.87 	&	10.08 	&	11.95 	&	5.08 	&	-1.09 	&	10.80 	&11.89 	\\
			education	&	-5.95 	&	-11.94 	&	0.05 	&	11.99 	&	-5.98 	&	-11.88 	&	0.26 	&	12.14 	&	-6.05 	&	-12.10 	&	0.00 	&	12.10 	&	-6.07 	&	-12.73 	&	0.54 	&	13.27 	\\
			group	&	-2.03 	&	-8.42 	&	4.35 	&	12.76 	&	-2.04 	&	-8.55 	&	4.20 	&	12.75 	&	-2.05 	&	-8.47 	&	4.37 	&	12.84 	&	-2.42 	&	-8.65 	&	3.85 	&12.50 	\\
			t1	&	-1.62 	&	-3.28 	&	0.03 	&	3.32 	&	$-1.62^{\ast}$ 	&	-3.76 	&	-0.25 	&	3.51 	&	$-1.62^{\ast}$	&	-3.22 	&	-0.03 	&	3.19 	&	$-1.13^{\ast}$ 	&	-1.75 	&	-0.01 	&	1.74 	\\
			t2	&	-1.62 	&	-3.71 	&	0.47 	&	4.18 	&	-1.63 	&	-4.26 	&	0.30 	&	4.56 	&	-1.64 	&	-3.71 	&	0.43 	&	4.15 	&	$-0.93^{\ast}$ 	&	-2.64 	&	-0.12 	&	2.52 	\\
			group*t1	&	-0.92 	&	-2.91 	&	1.07 	&	3.98 	&	-0.95 	&	-2.76 	&	1.41 	&	4.17 	&	-1.00 	&	-2.94 	&	0.94 	&3.89 	&	-1.43 	&	-2.07 	&	0.28 	&	2.35 	\\
			group*t2	&	$-2.73^{\ast}$ 	&	-5.39 	&	-0.08 	&	5.31 	&	-2.75 	&	-5.38 	&	0.24 	&	5.62 	&	$-2.78^{\ast}$ 	&	-5.43 	&	-0.13 	&	5.30 	&	$-3.24^{\ast}$ 	&	-5.38 	&	-2.14 	&	3.24 	\\
			\hline	 			
		\end{tabular*}
		\begin{tablenotes}	
			\item Note: \textbf{All the values of results are multiplied by $\bm{100}$}. GEEN: the naive GEE method; ELN: the naive EL method; LIN: \cite{Lin2018Analysis}'s estimating equation method;  Proposed: the proposed EL method; Coef: the estimate of regression coefficient; Lower: the lower bound of confidence interval; Upper: the upper bound of confidence interval; CL: the length of confidence interval; An asterisk ($^\ast$)  indicates that the effect is significant at the level of $\alpha=0.05$.
		\end{tablenotes}	
	\end{threeparttable}	
	
\end{sidewaystable}

\section{Conclusion and discussion}
\label{s:discussion}
In this paper, we propose a new method for  analysis of  longitudinal data with replicate covariate measurement errors based on the empirical likelihood estimator, where we  use the independence between replicate measurement errors to construct an unbiased auxiliary random vector.  When some elements in the original auxiliary random vector are functionally dependent or have some inner relationships, the reduced auxiliary random vector is introduced to define the profile empirical likelihood ratio function. The proposed method has the following advantages. Firstly, the proposed empirical likelihood estimator is asymptotically at least as efficient as the estimator of  \cite{Lin2018Analysis}, and under some moment conditions, the proposed estimator is strictly more efficient. Secondly, it provides a method that can deal with the problem where there are more than two replicate measurements at each assessment time.  Thirdly, it enjoys all the good properties of the empirical likelihood method. According to the simulation results, the proposed method is not sensitive to the measurement errors in the covariate. In addition, it is more efficient than \cite{Lin2018Analysis}'s method when the number of subjects is not too small. Due to the flexibility in accounting for the distributional difference of measurement errors in replicate measurements, we recommend using the proposed method in practical problems.

Further extension of the proposed method may be of interest.  One may consider how to extend this method to other models, such as the partially linear model and non-linear model. One may also consider how to improve the performance of the proposed method when the number of subjects is very small \citep{chen2008adjusted}.

\section*{Acknowledgements}
We would like to thank the Co-Editor, Dr. Byeong U. Park, an Associate Editor, and two referees for very helpful suggestions, which led to substantial improvements of the paper. We gratefully acknowledge Dr. Xuemei Sui of the University of South Carolina for providing the LEAN data set.  This work was  supported by the National Natural Science Foundation of China [11871164, 11671096, 11731011, 12071087].

	\thispagestyle{empty}
\bibliographystyle{apalike}
\bibliography{empirical_newest}

%\vfill
\clearpage

\appendix

\setcounter{theorem}{0}
\setcounter{equation}{0}
\setcounter{figure}{0}
\setcounter{table}{0}
\setcounter{section}{0}

\def\thesection{Appendix \Alph{section}}

\def\thetheorem{\Alph{section}.\arabic{theorem}}

\renewcommand{\theequation}{A\arabic{equation}}

%\renewcommand{\theequation}{A\arabic{equation}}
%\renewcommand{\thefigure}{A\arabic{figure}}
%\renewcommand{\thetable}{A\arabic{table}}
%%\renewcommand{\theassumption}{A\arabic{assumption}}
%\def\thesection{Appendix \Alph{section}}
%
%\def\thetheorem{\Alph{section}.\arabic{theorem}}
%
%\appendix

%\setcounter{theorem}{0}
%\def\thesection{Appendix \Alph{section}}

\section{Construction of the reduced  auxiliary random vector in a toy example}
\label{ss:construction}
Now we use a toy example to introduce how to construct the reduced auxiliary random vector.  Assume there are two covariates $X_{ij,1}$ and $X_{ij,2}$, where $X_{ij,1}$ is measured with error and there are three replicate measurements of $X_{ij,1}$. Denote the surrogate values of $X_{ij,1}$ as $W_{ij(k),1}$, $k=1,2,3$. We further assume the response variable and  covariates follow the following linear regression model:
\[
Y_{ij}=\beta_{00}+X_{ij,1}\beta_{01}+X_{ij,2}\beta_{02}+\varepsilon_{ij}, \quad i=1,\cdots,n, \quad j=1,\cdots,m_{i}.
\]
Therefore, $\bm{W}_{ij(k)}=(1,W_{ij(k),1},X_{ij,2} )^\top$ and $\bm{\beta}_{0}=(\beta_{00},\beta_{01},\beta_{02})^\top$.  According to the definition in Section~\ref{s:proposedmethod},
\[
\begin{aligned}
\bm{g}_{i}(\bm{\beta})&=\begin{pmatrix}\mathbf{W}_{i(1)}^\top\bm{\Sigma}_{i}^{-1}(\bm{Y}_{i}-\mathbf{W}_{i(2)}\bm{\beta}) \\ \mathbf{W}_{i(2)}^\top\bm{\Sigma}_{i}^{-1}(\bm{Y}_{i}-\mathbf{W}_{i(1)}\bm{\beta}) \\ \mathbf{W}_{i(1)}^\top\bm{\Sigma}_{i}^{-1}(\bm{Y}_{i}-\mathbf{W}_{i(3)}\bm{\beta}) \\ \mathbf{W}_{i(3)}^\top\bm{\Sigma}_{i}^{-1}(\bm{Y}_{i}-\mathbf{W}_{i(1)}\bm{\beta}) \\ \mathbf{W}_{i(2)}^\top\bm{\Sigma}_{i}^{-1}(\bm{Y}_{i}-\mathbf{W}_{i(3)}\bm{\beta}) \\  \mathbf{W}_{i(3)}^\top\bm{\Sigma}_{i}^{-1}(\bm{Y}_{i}-\mathbf{W}_{i(2)}\bm{\beta}) \end{pmatrix},
\end{aligned}
\]
where
\[
\begin{aligned}
\mathbf{W}_{i(k_{1})}^\top\bm{\Sigma}_{i}^{-1}(\bm{Y}_{i}-\mathbf{W}_{i(k_{2})}\bm{\beta})&=\begin{pmatrix}\mathbf{1}_{i}^\top\bm{\Sigma}_{i}^{-1}(\bm{Y}_{i}-\mathbf{W}_{i(k_{2})}\bm{\beta}) \\ \bm{W}_{i(k_{1}),1}^\top\bm{\Sigma}_{i}^{-1}(\bm{Y}_{i}-\mathbf{W}_{i(k_{2})}\bm{\beta}) \\ \bm{X}_{i,2}^\top\bm{\Sigma}_{i}^{-1}(\bm{Y}_{i}-\mathbf{W}_{i(k_{2})}\bm{\beta})  \end{pmatrix},\quad k_{1}\neq k_{2},
\end{aligned}
\]
with $\mathbf{1}_{i}$ being an $m_i$-dimensional vector with all the elements being one,  $\bm{X}_{i,2}=(X_{i1,2},\ldots,X_{im_i,2})^\top$, and $\bm{W}_{i(k_{1}),1}=(W_{i1(k_1),1},\ldots,W_{im_i(k_1),1})^\top$. It is obvious that for a fixed value of $k_{2}$, the values of $\mathbf{1}_{i}^\top\bm{\Sigma}_{i}^{-1}(\bm{Y}_{i}-\mathbf{W}_{i(k_{2})}\bm{\beta})$ are the same for different choices of $k_{1}$, it is also the case for $\bm{X}_{i,2}^\top\bm{\Sigma}_{i}^{-1}(\bm{Y}_{i}-\mathbf{W}_{i(k_{2})}\bm{\beta})$.  As a result,  the matrix  $\mathrm{E}\{\bm{g}_{i}(\bm{\beta}_{0})\bm{g}_{i}(\bm{\beta}_{0})^\top\}$ is not  invertible and the estimated value of $\bm{\beta}_{0}$ is unavailable. To solve this problem, we need to eliminate the duplicate elements in $\bm{g}_{i}(\bm{\beta})$. Besides, there  is also a potential inner relationship among $\big\{\bm{W}_{i(k_{1}),1}^\top\bm{\Sigma}_{i}^{-1}(\bm{Y}_{i}-\mathbf{W}_{i(k_{2})}\bm{\beta}),k_{1}\neq k_{2}\big\}$. Denote
\[
\begin{aligned}
\tilde{\bm{g}}_{i}(\bm{\beta})&=\begin{pmatrix}\bm{W}_{i(1),1}^\top\bm{\Sigma}_{i}^{-1}(\bm{Y}_{i}-\mathbf{W}_{i(2)}\bm{\beta}) \\\bm{W}_{i(2),1}^\top\bm{\Sigma}_{i}^{-1}(\bm{Y}_{i}-\mathbf{W}_{i(1)}\bm{\beta}) \\ \bm{W}_{i(1),1}^\top\bm{\Sigma}_{i}^{-1}(\bm{Y}_{i}-\mathbf{W}_{i(3)}\bm{\beta}) \\ \bm{W}_{i(3),1}^\top\bm{\Sigma}_{i}^{-1}(\bm{Y}_{i}-\mathbf{W}_{i(1)}\bm{\beta}) \\ \bm{W}_{i(2),1}^\top\bm{\Sigma}_{i}^{-1}(\bm{Y}_{i}-\mathbf{W}_{i(3)}\bm{\beta}) \\ \bm{W}_{i(3),1}^\top\bm{\Sigma}_{i}^{-1}(\bm{Y}_{i}-\mathbf{W}_{i(2)}\bm{\beta}) \end{pmatrix},
\end{aligned}
\]
and $\mathbf{A}=\mathrm{E}\{\tilde{\bm{g}}_{i}(\bm{\beta}_{0})\tilde{\bm{g}}_{i}(\bm{\beta}_{0})^\top\}$. Let $\bm{A}_{j}$ denote the $j$th row vector of $\mathbf{A}$, $j=1,\cdots,6$. By calculation, we can find that $\bm{A}_{1}+\bm{A}_{4}+\bm{A}_{5}=\bm{A}_{2}+\bm{A}_{3}+\bm{A}_{6}$. Therefore, $\mathbf{A}$ is not invertible. Through performing row transformation on the matrix $\mathrm{E}\{\bm{g}_{i}(\bm{\beta}_{0})\bm{g}_{i}(\bm{\beta}_{0})^\top\}$, $\mathbf{A}$ can be one block of the transformed matrix. As a result, if  $\mathbf{A}$ is not invertible, then $\mathrm{E}\{\bm{g}_{i}(\bm{\beta}_{0})\bm{g}_{i}(\bm{\beta}_{0})^\top\}$ will also be not invertible. Therefore, to make $\mathrm{E}\{\bm{g}_{i}(\bm{\beta}_{0})\bm{g}_{i}(\bm{\beta}_{0})^\top\}$ invertible, we  need to eliminate  one element in $\tilde{\bm{g}}_{i}(\bm{\beta})$ from $\tilde{\bm{g}}_{i}(\bm{\beta})$. After eliminating the duplicate elements in $\bm{g}_{i}(\bm{\beta})$ and one element in $\tilde{\bm{g}}_{i}(\bm{\beta})$, we can get the reduced auxiliary random vector $\bm{g}_{i}^{*}(\bm{\beta})$ and
\[
\begin{aligned}
\bm{g}_{i}^{*}(\bm{\beta})&=\begin{pmatrix}\mathbf{1}_{i}^\top\bm{\Sigma}_{i}^{-1}(\bm{Y}_{i}-\mathbf{W}_{i(2)}\bm{\beta}) \\

\bm{W}_{i(1),1}^\top\bm{\Sigma}_{i}^{-1}(\bm{Y}_{i}-\mathbf{W}_{i(2)}\bm{\beta}) \\

\bm{X}_{i,2}^\top\bm{\Sigma}_{i}^{-1}(\bm{Y}_{i}-\mathbf{W}_{i(2)}\bm{\beta})\\

\mathbf{1}_{i}^\top\bm{\Sigma}_{i}^{-1}(\bm{Y}_{i}-\mathbf{W}_{i(1)}\bm{\beta}) \\

\bm{W}_{i(2),1}^\top\bm{\Sigma}_{i}^{-1}(\bm{Y}_{i}-\mathbf{W}_{i(1)}\bm{\beta}) \\

\bm{X}_{i,2}^\top\bm{\Sigma}_{i}^{-1}(\bm{Y}_{i}-\mathbf{W}_{i(1)}\bm{\beta})\\

\mathbf{1}_{i}^\top\bm{\Sigma}_{i}^{-1}(\bm{Y}_{i}-\mathbf{W}_{i(3)}\bm{\beta}) \\

\bm{W}_{i(1),1}^\top\bm{\Sigma}_{i}^{-1}(\bm{Y}_{i}-\mathbf{W}_{i(3)}\bm{\beta}) \\

\bm{X}_{i,2}^\top\bm{\Sigma}_{i}^{-1}(\bm{Y}_{i}-\mathbf{W}_{i(3)}\bm{\beta})\\

\bm{W}_{i(3),1}^\top\bm{\Sigma}_{i}^{-1}(\bm{Y}_{i}-\mathbf{W}_{i(1)}\bm{\beta}) \\
\bm{W}_{i(2),1}^\top\bm{\Sigma}_{i}^{-1}(\bm{Y}_{i}-\mathbf{W}_{i(3)}\bm{\beta})
\end{pmatrix}.
\end{aligned}
\]
In this example, the number of elements in $\bm{g}_{i}^{*}(\bm{\beta})$ is $11$.

\section{Efficiency of the empirical likelihood estimator in the toy example}
\label{ss:efficiency}
In the toy example given in \ref{ss:construction}, the estimating equation for \cite{Lin2018Analysis}'s method  is
\[
\bm{U}(\bm{\beta})=\sum_{i=1}^{n}\bm{U}_{i}(\bm{\beta})=\sum_{i=1}^{n}\sum_{k_{1}\neq k_{2}}\mathbf{W}_{i(k_{1})}^\top\bm{\Sigma}_{i}^{-1}(\bm{Y}_{i}-\mathbf{W}_{i(k_{2})}\bm{\beta})=\bm{0},
\]
where $k_{1}$ and $k_{2}$ are the elements of $\{1,2,3\}$. Let 
$$\bm{h}_{i}(\bm{\beta})=[\{\bm{g}_{i}^{*}(\bm{\beta})\}^\top, \bm{W}_{i(3),1}^\top\bm{\Sigma}_{i}^{-1}(\bm{Y}_{i}-\mathbf{W}_{i(2)}\bm{\beta})]^\top,$$
 then $\bm{U}_{i}(\bm{\beta})=\mathbf{B}_{1}\bm{h}_{i}(\bm{\beta})$, where
\[
\begin{aligned}
\mathbf{B}_{1}&=\left(
\begin{array}{cccccccccccc}
2 & 0 & 0&2&0&0&2&0&0&0&0&0\\
0 & 1 & 0&0&1&0&0&1&0&1&1&1\\
0&0&2&0&0&2&0&0&2&0&0&0\\
\end{array}
\right).
\end{aligned}
\]

For any estimator $\hat{\theta}$, we denote the asymptotic variance of $\hat{\theta}$ as $\mathrm{asyVar}(\hat{\theta})$. According to Theorem 2 in  \cite{Lin2018Analysis}, the asymptotic variance of  $\hat{\bm{\beta}}_{LIN}$ is
\begin{multline*}
	\mathrm{asyVar}(\hat{\bm{\beta}}_{LIN})\\
	=\left[\sum_{i=1}^{n}\mathrm{E}\Big\{\frac{\partial \bm{U}_{i}(\bm{\beta}_{0})}{\partial \bm{\beta}^\top}\Big\}\right]^{-1}\left[\sum_{i=1}^{n}\mathrm{E}\big\{\bm{U}_{i}(\bm{\beta}_{0})\bm{U}_{i}(\bm{\beta}_{0})^\top\big\}\right]\left[\sum_{i=1}^{n}\mathrm{E}\Big\{\frac{\partial \bm{U}_{i}(\bm{\beta}_{0})}{\partial \bm{\beta}^\top}\Big\}\right]^{-1.\top}\\
	=\left[\sum_{i=1}^{n}\mathbf{B}_{1}\mathrm{E}\Big\{\frac{\partial \bm{h}_{i}(\bm{\beta}_{0})}{\partial \bm{\beta}^\top}\Big\}\right]^{-1}\left[\sum_{i=1}^{n}\mathbf{B}_{1}\mathrm{E}\big\{\bm{h}_{i}(\bm{\beta}_{0})\bm{h}_{i}(\bm{\beta}_{0})^\top\big\}\mathbf{B}_{1}^\top\right]\left[\sum_{i=1}^{n}\mathbf{B}_{1}\mathrm{E}\Big\{\frac{\partial \bm{h}_{i}(\bm{\beta}_{0})}{\partial \bm{\beta}^\top}\Big\}\right]^{-1.\top}.
\end{multline*}
By calculation,
\[
\mathrm{E}\left\{\frac{\partial \bm{h}_{i}(\bm{\beta}_{0})}{\partial \bm{\beta}^\top}\right\}=\mathbf{B}_{2}\mathrm{E}\left\{\frac{\partial \bm{g}_{i}^{*}(\bm{\beta}_{0})}{\partial \bm{\beta}^\top}\right\},
\]
and
\[
\mathrm{E}\big\{\bm{h}_{i}(\bm{\beta}_{0})\bm{h}_{i}(\bm{\beta}_{0})^\top\big\}=\mathbf{B}_{2}\mathrm{E}\big\{\bm{g}_{i}^{*}(\bm{\beta}_{0})\bm{g}_{i}^{*}(\bm{\beta}_{0})^\top\big\}\mathbf{B}_{2}^\top,
\]
where
\[
\mathbf{B}_{2}=\begin{pmatrix}
\mathbf{I}_{11\times 11}\\ \bm{B}_{3}^\top
\end{pmatrix},
\]
with $\bm{B}_{3}=(0,1,0,0,-1,0,0,-1,0,1,1)^\top$ and $\mathbf{I}_{p\times p}$ being the $p\times p$ identity matrix for any positive integer $p$. Let $\mathbf{B}=\mathbf{B}_{1}\mathbf{B}_{2}$, then
\begin{multline*}
	\mathrm{asyVar}(\hat{\bm{\beta}}_{LIN})\\=\left[\sum_{i=1}^{n}\mathbf{B}\mathrm{E}\Big\{\frac{\partial \bm{g}_{i}^{*}(\bm{\beta}_{0})}{\partial \bm{\beta}^\top}\Big\}\right]^{-1}\left[\sum_{i=1}^{n}\mathbf{B}\mathrm{E}\big\{\bm{g}_{i}^{*}(\bm{\beta}_{0})\bm{g}_{i}^{*}(\bm{\beta}_{0})^\top\big\}\mathbf{B}^\top\right]\left[\sum_{i=1}^{n}\mathbf{B}\mathrm{E}\Big\{\frac{\partial \bm{g}_{i}^{*}(\bm{\beta}_{0})}{\partial \bm{\beta}^\top}\Big\}\right]^{-1.\top}.
\end{multline*}
Therefore, the asymptotic variance of $\hat{\bm{\beta}}_{LIN}$ is equal to the asymptotic variance of the estimator obtained from the estimating equation $\sum_{i=1}^{n}\mathbf{B}\bm{g}_{i}^{*}(\bm{\beta})=\bm{0}$. According to Theorem 3.6 in  \cite{owen2001empirical}, the asymptotic variance of the empirical likelihood estimator $\hat{\bm{\beta}}$ is at least as small as that of any estimator obtained from $\sum_{i=1}^{n}\mathbf{C}\bm{g}_{i}^{*}(\bm{\beta})=\bm{0}$, where $\mathbf{C}$ is an arbitrary $p\times q$ matrix.  Therefore, the asymptotic variance of the empirical likelihood estimator $\hat{\bm{\beta}}$ is at least as small as that of $\hat{\bm{\beta}}_{LIN}$.

\section{Comparison of the asymmetric variances}
\label{ss:Comparison}

Without loss of generality, we consider the case where there are two replicate measurements for the covariate. For \cite{Lin2018Analysis}'s method,  the estimating equation  is
\[
\bm{U}(\bm{\beta})=\sum_{i=1}^{n}\bm{U}_{i}(\bm{\beta})=\sum_{i=1}^{n}\mathbf{W}_{i(1)}^\top\bm{\Sigma}_{i}^{-1}(\bm{Y}_{i}-\mathbf{W}_{i(2)}\bm{\beta})\!+\!\mathbf{W}_{i(2)}^\top\bm{\Sigma}_{i}^{-1}(\bm{Y}_{i}-\mathbf{W}_{i(1)}\bm{\beta})=\bm{0}.
\]
The asymptotic variance of  $\hat{\bm{\beta}}_{LIN}$ is
\begin{multline}
	\mathrm{asyVar}(\hat{\bm{\beta}}_{LIN})\\
	=\left[\sum_{i=1}^{n}\mathrm{E}\Big\{\frac{\partial \bm{U}_{i}(\bm{\beta}_{0})}{\partial \bm{\beta}^\top}\Big\}\right]^{-1}\left[\sum_{i=1}^{n}\mathrm{E}\big\{\bm{U}_{i}(\bm{\beta}_{0})\bm{U}_{i}(\bm{\beta}_{0})^\top\big\}\right]\left[\sum_{i=1}^{n}\mathrm{E}\Big\{\frac{\partial \bm{U}_{i}(\bm{\beta}_{0})}{\partial \bm{\beta}^\top}\Big\}\right]^{-1.\top}\\
	=\left\{\sum_{i=1}^{n}\mathrm{E}\big(\mathbf{X}_{i}^\top\Sigma_{i}^{-1}\mathbf{X}_{i}\big)\right\}^{-1}\left[\frac{1}{4}\sum_{i=1}^{n}\mathrm{E} \{\bm{U}_{i}(\bm{\beta}_{0})\bm{U}_{i}(\bm{\beta}_{0})^\top\}\right]\left\{\sum_{i=1}^{n}\mathrm{E}\big(\mathbf{X}_{i}^\top\Sigma_{i}^{-1}\mathbf{X}_{i}\big)\right\}^{-1},
	\label{asymvarlin1}
\end{multline}
where
$\sum_{i=1}^{n}\mathrm{E} \{\bm{U}_{i}(\bm{\beta}_{0})\bm{U}_{i}(\bm{\beta}_{0})^\top\}/4=\mathbf{D}+(\mathbf{E}_{1}+\mathbf{E}_{2}+\mathbf{F}_{1}+\mathbf{F}_{2}+\mathbf{G}_{1}+\mathbf{G}_{2}+\mathbf{G}_{3}+\mathbf{G}_{3}^\top )/4$ with $\mathbf{D}=\sum_{i=1}^{n}\mathrm{E}(\mathbf{X}_{i}^\top\Sigma_{i}^{-1}\mathbf{X}_{i})$,  $\mathbf{E}_{1}=\sum_{i=1}^{n}\mathrm{E}\big\{\mathbf{X}_{i}^\top\Sigma_{i}^{-1}\mathrm{cov}(\bm{\delta}_{i(2)}\bm{\beta}_{0})\Sigma_{i}^{-1}\mathbf{X}_{i}\big\}$, $\mathbf{E}_{2}=\sum_{i=1}^{n}\mathrm{E}\big\{\mathbf{X}_{i}^\top\Sigma_{i}^{-1}\mathrm{cov}(\bm{\delta}_{i(1)}\bm{\beta}_{0})\Sigma_{i}^{-1}\mathbf{X}_{i}\big\}$,  $\mathbf{F}_{1}=\sum_{i=1}^{n}\mathrm{E}(\bm{\delta}_{i(1)}^\top\Sigma_{i}^{-1}\bm{\delta}_{i(1)})$,  $\mathbf{F}_{2}=\sum_{i=1}^{n}\mathrm{E}(\bm{\delta}_{i(2)}^\top\Sigma_{i}^{-1}\bm{\delta}_{i(2)})$,  $\mathbf{G}_{1}=\sum_{i=1}^{n}\mathrm{cov}(\bm{\delta}_{i(1)}^\top\Sigma_{i}^{-1}\bm{\delta}_{i(2)}\bm{\beta}_{0})$, $\mathbf{G}_{2}=\sum_{i=1}^{n}\mathrm{cov}(\bm{\delta}_{i(2)}^\top\Sigma_{i}^{-1}\bm{\delta}_{i(1)}\bm{\beta}_{0})$,
and $\mathbf{G}_{3}=\sum_{i=1}^{n}\mathrm{E}(\bm{\delta}_{i(1)}^\top\Sigma_{i}^{-1}\bm{\delta}_{i(2)}\bm{\beta}_{0}\bm{\beta}_{0}^\top\bm{\delta}_{i(1)}^\top\Sigma_{i}^{-1}\bm{\delta}_{i(2)})$.

For the proposed empirical likelihood method, the auxiliary random vector  is
\[
\begin{aligned}
\bm{g}_{i}(\bm{\beta})&=\begin{pmatrix}\mathbf{W}_{i(1)}^\top\bm{\Sigma}_{i}^{-1}(\bm{Y}_{i}-\mathbf{W}_{i(2)}\bm{\beta}) \\ \mathbf{W}_{i(2)}^\top\bm{\Sigma}_{i}^{-1}(\bm{Y}_{i}-\mathbf{W}_{i(1)}\bm{\beta})  \end{pmatrix}.
\end{aligned}
\]
According to Theorem~\ref{theo:normality}, the asymptotic variance of  $\hat{\bm{\beta}}$ is
\begin{equation}
	\begin{aligned}
		&\mathrm{asyVar}(\hat{\bm{\beta}})\\
		=&\left[\left\{\sum_{i=1}^{n}\mathrm{E}\Big(\frac{\partial \bm{g}_{i}(\bm{\beta}_{0})}{\partial \bm{\beta}^\top}\Big)\right\}^\top\left\{\sum_{i=1}^{n}\mathrm{E}\big(\bm{g}_{i}(\bm{\beta}_{0})\bm{g}_{i}(\bm{\beta}_{0})^\top\big)\right\}^{-1}\left\{\sum_{i=1}^{n}\mathrm{E}\Big(\frac{\partial \bm{g}_{i}(\bm{\beta}_{0})}{\partial \bm{\beta}^\top}\Big)\right\}\right]^{-1}\\
		=&\left\{\sum_{i=1}^{n}\mathrm{E}\big(\mathbf{X}_{i}^\top\Sigma_{i}^{-1}\mathbf{X}_{i}\big)\right\}^{-1}\left[\begin{pmatrix} \mathbf{I}_{p\times p} &  \mathbf{I}_{p\times p} \end{pmatrix}\left\{\sum_{i=1}^{n}\mathrm{E}\big(\bm{g}_{i}(\bm{\beta}_{0})\bm{g}_{i}(\bm{\beta}_{0})^\top\big)\right\}^{-1}\begin{pmatrix} \mathbf{I}_{p\times p} \\  \mathbf{I}_{p\times p} \end{pmatrix} \right]^{-1}\\
		&\times\left\{\sum_{i=1}^{n}\mathrm{E}\big(\mathbf{X}_{i}^\top\Sigma_{i}^{-1}\mathbf{X}_{i}\big)\right\}^{-1},
	\end{aligned}
	\label{asymvarpro}
\end{equation}
where
\[
\sum_{i=1}^{n}\mathrm{E}\big\{\bm{g}_{i}(\bm{\beta}_{0})\bm{g}_{i}(\bm{\beta}_{0})^\top\big\}=\begin{pmatrix} \mathbf{D}+\mathbf{E}_{1}+\mathbf{F}_{1}+\mathbf{G}_{1} & \mathbf{D}+\mathbf{G}_{3}\\ \mathbf{D}+\mathbf{G}_{3}^\top& \mathbf{D}+\mathbf{E}_{2}+\mathbf{F}_{2}+\mathbf{G}_{2} \end{pmatrix}.
\]

By calculation,
\begin{multline*}
	\Delta\triangleq 
	\frac{1}{4}\sum_{i=1}^{n}\mathrm{E}\big\{\bm{U}_{i}(\bm{\beta}_{0})\bm{U}_{i}(\bm{\beta}_{0})^\top\big\}-\left[\begin{pmatrix} \mathbf{I}_{p\times p} &  \mathbf{I}_{p\times p} \end{pmatrix}\left\{\sum_{i=1}^{n}\mathrm{E} \big(\bm{g}_{i}(\bm{\beta}_{0})\bm{g}_{i}(\bm{\beta}_{0})^\top\big)\right\}^{-1}\begin{pmatrix} \mathbf{I}_{p\times p} \\  \mathbf{I}_{p\times p} \end{pmatrix} \right]^{-1}\\
	=\frac{1}{4}(\mathbf{E}_{1}+\mathbf{E}_{2}+\mathbf{F}_{1}+\mathbf{F}_{2}+\mathbf{G}_{1}+\mathbf{G}_{2}+\mathbf{G}_{3}+\mathbf{G}_{3}^\top )-\mathbf{E}_{1}-\mathbf{F}_{1}-\mathbf{G}_{1}\\
	+(\mathbf{E}_{1}+\mathbf{F}_{1}+\mathbf{G}_{1}-\mathbf{G}_{3})(\mathbf{E}_{1}+\mathbf{F}_{1}+\mathbf{G}_{1}-\mathbf{G}_{3}+\mathbf{E}_{2}+\mathbf{F}_{2}+\mathbf{G}_{2}-\mathbf{G}_{3}^\top)^{-1}(\mathbf{E}_{1}+\mathbf{F}_{1}+\mathbf{G}_{1}-\mathbf{G}_{3}^\top).
\end{multline*}
Denote $\mathbf{H}_{1}=\mathbf{E}_{1}+\mathbf{F}_{1}+\mathbf{G}_{1}-\mathbf{G}_{3}$ and  $\mathbf{H}_{2}=\mathbf{E}_{2}+\mathbf{F}_{2}+\mathbf{G}_{2}-\mathbf{G}_{3}^\top$, then it is easy to verify that if $\mathbf{H}_{1}=\mathbf{H}_{2}$, then $\Delta=0$. Furthermore, according to~(\ref{asymvarlin1}) and~(\ref{asymvarpro}), $\mathrm{asyVar}(\hat{\bm{\beta}}_{LIN})=\mathrm{asyVar}(\hat{\bm{\beta}})$. Note that the condition of $\mathbf{H}_{1}=\mathbf{H}_{2}$ can be easily satisfied if the two replicate measurements for the  same covariate follow the same distribution.

However, the condition of $\mathbf{H}_{1}=\mathbf{H}_{2}$ may be violated if the distributions of two  replicate measurements for the same covariate are different.  If $\mathbf{H}_{1}\neq\mathbf{H}_{2}$, denote $\mathbf{H}_{1}-\mathbf{H}_{2}$ as $\mathbf{\Lambda}$. For convenience, assume $\mathbf{G}_{3}=\mathbf{G}_{3}^\top$, then
\[
\begin{aligned}
\Delta(\mathbf{\Lambda})&=\mathbf{H}_{1}(\mathbf{H}_{1}+\mathbf{H}_{2})^{-1}\mathbf{H}_{1}+\frac{1}{4}\mathbf{H}_{2}-\frac{3}{4}\mathbf{H}_{1}\\
&=\mathbf{H}_{1}(2\mathbf{H}_{1}-\mathbf{\Lambda})^{-1}\mathbf{H}_{1}-\frac{1}{2}\mathbf{H}_{1}-\frac{1}{4}\mathbf{\Lambda}.
\end{aligned}
\]
Obviously, $\Delta(\bm{0})=\bm{0}$. Besides, for any $\mathbf{\Lambda}\neq \bm{0}$, based on the  Woodbury matrix identity,
\[
\begin{aligned}
\Delta(\mathbf{\Lambda})-\Delta(\bm{0})&=\mathbf{H}_{1}\big\{(2\mathbf{H}_{1}-\mathbf{\Lambda})^{-1}-(2\mathbf{H}_{1})^{-1}\big\}\mathbf{H}_{1}-\frac{1}{4}\mathbf{\Lambda}\\
&=\frac{1}{4}\big\{\mathbf{\Lambda}^{-1}-(2\mathbf{H}_{1})^{-1}\big\}^{-1}-\frac{1}{4}\mathbf{\Lambda}\\
&=\frac{1}{4}\mathbf{\Lambda}(2\mathbf{H}_{1}-\mathbf{\Lambda})^{-1}\mathbf{\Lambda}\\
&=\frac{1}{4}\mathbf{\Lambda}(\mathbf{H}_{1}+\mathbf{H}_{2})^{-1}\mathbf{\Lambda}.
\end{aligned}
\]

Note that
\[
\mathbf{H}_{1}+\mathbf{H}_{2}=\begin{pmatrix} \mathbf{I}_{p\times p} &  -\mathbf{I}_{p\times p} \end{pmatrix}\left\{\sum_{i=1}^{n}\mathrm{E} (\bm{g}_{i}(\bm{\beta}_{0})\bm{g}_{i}(\bm{\beta}_{0})^\top)\right\}^{-1}\begin{pmatrix} \mathbf{I}_{p\times p} \\  -\mathbf{I}_{p\times p} \end{pmatrix},
\]
then
\[
\begin{aligned}
&\mathrm{asyVar}(\hat{\bm{\beta}}_{LIN})-\mathrm{asyVar}(\hat{\bm{\beta}})\\
=&\left\{\sum_{i=1}^{n}\mathrm{E}\big(\mathbf{X}_{i}^\top\Sigma_{i}^{-1}\mathbf{X}_{i}\big)\right\}^{-1}\Delta(\mathbf{\Lambda})\left\{\sum_{i=1}^{n}\mathrm{E}\big(\mathbf{X}_{i}^\top\Sigma_{i}^{-1}\mathbf{X}_{i}\big)\right\}^{-1}\\
=&\frac{1}{4}\left\{\sum_{i=1}^{n}\mathrm{E}\big(\mathbf{X}_{i}^\top\Sigma_{i}^{-1}\mathbf{X}_{i}\big)\right\}^{-1}\mathbf{\Lambda}\left[\begin{pmatrix} \mathbf{I}_{p\times p} &  -\mathbf{I}_{p\times p} \end{pmatrix}\left\{\sum_{i=1}^{n}\mathrm{E} \big(\bm{g}_{i}(\bm{\beta}_{0})\bm{g}_{i}(\bm{\beta}_{0})^\top\big)\right\}^{-1}\begin{pmatrix} \mathbf{I}_{p\times p} \\  -\mathbf{I}_{p\times p} \end{pmatrix}\right]^{-1}\\
&\times\mathbf{\Lambda}\left\{\sum_{i=1}^{n}\mathrm{E}\big(\mathbf{X}_{i}^\top\Sigma_{i}^{-1}\mathbf{X}_{i}\big)\right\}^{-1}.
\end{aligned}
\]
Since $\sum_{i=1}^{n}\mathrm{E}\{\bm{g}_{i}(\bm{\beta}_{0})\bm{g}_{i}(\bm{\beta}_{0})^\top\}$ is a positive definite matrix, then $\mathrm{asyVar}(\hat{\bm{\beta}}_{LIN})-\mathrm{asyVar}(\hat{\bm{\beta}})>\bm{0}$.

In summary,  if the model satisfies the condition of $\mathbf{H}_{1}=\mathbf{H}_{2}$, then the asymptotic variance of \cite{Lin2018Analysis}'s estimator is the same as  that of the proposed estimator. This condition can be satisfied naturally if the two replicate measurements for the same covariate follow the same distribution. However, if this condition is violated, the asymptotic variance of \cite{Lin2018Analysis}'s estimator becomes larger than that of the proposed estimator. In general, the proposed estimator is asymptotically more efficient than \cite{Lin2018Analysis}'s estimator. This conclusion can be easily extended to the general case where there are more than two replicate measurements.

\section{ Proof of asymptotic properties}
\label{ss:proofofasym}
\subsection{Lemmas}
In order to prove Theorems~\ref{theo:normality},~\ref{theo:chisquaretrue},~\ref{theorem:fulltest}, and~\ref{theorem:profiletest}, we first introduce the following two lemmas.

\begin{lemma}
	\label{lemma:A1}
	Assuming that conditions (R.1)--(R.5) hold, we have
	\[
	\frac{1}{\sqrt{n}}\sum_{i=1}^{n}\bm{g}_{i}^{*}(\bm{\beta}_{0})\rightsquigarrow\mathcal{N}(\bm{0},\mathbf{M}).
	\]
\end{lemma}

\begin{proof}
	We first prove that  $\sum_{i=1}^{n}\mathrm{E}\|\bm{g}_{i}^{*}(\bm{\beta}_{0})/\sqrt{n}\|^{3}\rightarrow 0$. Note that $$\sum_{i=1}^{n}\mathrm{E}\|\bm{g}_{i}^{*}(\bm{\beta}_{0})/\sqrt{n}\|^{3}=\sum_{i=1}^{n}n^{-3/2}\mathrm{E}\|\bm{g}_{i}^{*}(\bm{\beta}_{0})\|^{3},$$
	 and
	\[
	\begin{aligned}
	\mathrm{E}\|\bm{g}_{i}^{*}(\bm{\beta}_{0})\|^{3}&\leq \mathrm{E}\|\bm{g}_{i}(\bm{\beta}_{0})\|^{3}=
	\mathrm{E}\left[\left\{\sum_{k_{1}\neq k_{2}}\|\mathbf{W}_{i(k_{1})}^\top\Sigma_{i}^{-1}(\bm{Y}_{i}-\mathbf{W}_{i(k_{2})}\bm{\beta}_{0})\|^{2}\right\}^{\frac{3}{2}}\right]\\
	&\leq \big\{K(K-1)\big\}^\frac{1}{2}\sum_{k_{1}\neq k_{2}}\mathrm{E}\big\|\mathbf{W}_{i(k_{1})}^\top\Sigma_{i}^{-1}(\bm{Y}_{i}-\mathbf{W}_{i(k_{2})}\bm{\beta}_{0})\big\|^{3}\\
	&=\big\{K(K-1)\big\}^\frac{1}{2}\sum_{k_{1}\neq k_{2}}\mathrm{E}\big\|(\mathbf{X}_{i}+\bm{\xi}_{i(k_{1})})^\top\Sigma_{i}^{-1}(\bm{\varepsilon}_{i}-\bm{\xi}_{i(k_{2})}\bm{\beta}_{0})\big\|^{3}\\
	&\leq \big\{K(K-1)\big\}^\frac{1}{2}\sum_{k_{1}\neq k_{2}}\mathrm{E}\big(\|\mathbf{X}_{i}+\bm{\xi}_{i(k_{1})}\|^{3}\|\Sigma_{i}^{-1}\|^{3}\|\bm{\varepsilon}_{i}-\bm{\xi}_{i(k_{2})}\bm{\beta}_{0}\|^{3}\big).
	\end{aligned}
	\]
	By conditions (R.1)--(R.4) and the Markov's inequality, we have $\mathrm{E}\|\bm{g}_{i}^{*}(\bm{\beta}_{0})\|^{3}<\infty$.  Furthermore, $\sum_{i=1}^{n}\mathrm{E}\|\bm{g}_{i}^{*}(\bm{\beta}_{0})/\sqrt{n}\|^{3}\rightarrow 0$.
	Based on condition (R.2),  $\mathrm{E}\{\bm{g}_{i}^{*}(\bm{\beta}_{0})/\sqrt{n}\}=\bm{0}$. By the law of large numbers and condition (R.5), $\sum_{i=1}^{n}\mathrm{cov}\{\bm{g}_{i}^{*}(\bm{\beta}_{0})/\sqrt{n}\}\rightarrow \mathbf{M}$. Therefore, according to Lyapunov central limit theorem, we have $\sum_{i=1}^{n}\bm{g}_{i}^{*}(\bm{\beta}_{0})/\sqrt{n}\rightsquigarrow\mathcal{N}(\bm{0},\mathbf{M})$.
\end{proof}

\begin{lemma}
	\label{lemma:A2}
	Assuming that conditions (R.1)--(R.5) hold, we have
	\begin{enumerate}
	\item [(a)] $\max_{1\leq i\leq n}\|\bm{g}_{i}^{*}(\bm{\beta}_{0})\|=o_{p}(n^{1/2})$,
	\item [(b)] $\sum_{i=1}^{n}\|\bm{g}_{i}^{*}(\bm{\beta}_{0})\|^{3}/n=o_{p}(n^{1/2})$, and
	\item [(c)] $\|\bm{\lambda}(\bm{\beta}_{0})\|=O_{p}(n^{-1/2})$.
			\end{enumerate}

\end{lemma}

\begin{proof}
	As  pointed out in the proof of Lemma~\ref{lemma:A1}, there exists a positive constant $\delta$ such that  $\mathrm{E}\|\bm{g}_{i}^{*}(\bm{\beta}_{0})\|^{2+\delta}<\infty$, therefore, $\mathrm{E}\|\bm{g}_{i}^{*}(\bm{\beta}_{0})\|^{2}<\infty$.  According to the Markov's inequality, $\sum_{i=1}^{n}{\rm P}(\|\bm{g}_{i}^{*}(\bm{\beta}_{0})\|^{2}>n)\leq \sum_{i=1}^{n}\mathrm{E}\|\bm{g}_{i}^{*}(\bm{\beta}_{0})\|^{2}/n<\infty$. Hence, by the Borel-Cantelli Lemma, $\|\bm{g}_{i}^{*}(\bm{\beta}_{0})\|>n^{1/2}$ finitely often with probability $1$. Let $\max_{1\leq i\leq n}\|\bm{g}_{i}^{*}(\bm{\beta}_{0})\|=Z_{n}$, then $Z_{n}>n^{1/2}$ finitely often. By the same argument, $Z_{n}>An^{1/2}$ finitely often for any $A>0$. Therefore, $\limsup_{n\rightarrow\infty}Z_{n}n^{-1/2}\leq A$ with probability $1$. So $Z_{n}=\max_{1\leq i\leq n}\|\bm{g}_{i}^{*}(\bm{\beta}_{0})\|=o_{p}(n^{1/2})$.
	
	Note that 
	$$\sum_{i=1}^{n}\|\bm{g}_{i}^{*}(\bm{\beta}_{0})\|^{3}/n=\sum_{i=1}^{n}\|\bm{g}_{i}^{*}(\bm{\beta}_{0})\|\|\bm{g}_{i}^{*}(\bm{\beta}_{0})\|^{2}/n\leq Z_{n}\times\sum_{i=1}^{n}\|\bm{g}_{i}^{*}(\bm{\beta}_{0})\|^{2}/n.$$ 
	By the law of large numbers, $\sum_{i=1}^{n}\|\bm{g}_{i}^{*}(\bm{\beta}_{0})\|^{2}/n-\sum_{i=1}^{n}\mathrm{E}\|\bm{g}_{i}^{*}(\bm{\beta}_{0})\|^{2}/n\stackrel{P}{\rightarrow}0$. Since $\mathrm{E}\|\bm{g}_{i}^{*}(\bm{\beta}_{0})\|^{2}<\infty$, then $\sum_{i=1}^{n}\|\bm{g}_{i}^{*}(\bm{\beta}_{0})\|^{2}/n=O_{p}(1)$. Furthermore, $\sum_{i=1}^{n}\|\bm{g}_{i}^{*}(\bm{\beta}_{0})\|^{3}/n=o_{p}(n^{1/2})$.
	
	For simplicity, we denote $\bm{\lambda}(\bm{\beta}_{0})$ as $\bm{\lambda}_{0}$. Write $\bm{\lambda}_{0}=\rho\bm{\theta}$, where $\rho\geq 0$ and $\|\bm{\theta}\|=1$. Since $\pi_{i}=1/[n\{1+\bm{\lambda}_{0}^\top\bm{g}_{i}^{*}(\bm{\beta}_{0})\}]$, $\pi_{i}\geq 0$ and $\sum_{i=1}^{n}\pi_{i}=1$, we have $1+\bm{\lambda}_{0}^\top\bm{g}_{i}^{*}(\bm{\beta}_{0})>0$. As we know, $\bm{\lambda}_{0}$ should satisfy the condition of  $\bm{0}=\sum_{i=1}^{n}\bm{g}_{i}^{*}(\bm{\beta}_{0})/[n\{1+\bm{\lambda}_{0}^\top\bm{g}_{i}^{*}(\bm{\beta}_{0})\}]\triangleq \bm{\ell}(\bm{\lambda}_{0})$. Besides,
	\begin{equation}
		\begin{aligned}
			0=&\|\bm{\ell}(\rho\bm{\theta})\|=\|\bm{\ell}(\rho\bm{\theta})\|\|\bm{\theta}\|\geq |\bm{\theta}^\top\bm{\ell}(\rho\bm{\theta})|\\
			=&\frac{1}{n}\bigg|\bm{\theta}^\top\Big\{\sum_{i=1}^{n}\bm{g}_{i}^{*}(\bm{\beta}_{0})-\rho \sum_{i=1}^{n}\frac{\bm{g}_{i}^{*}(\bm{\beta}_{0})\bm{\theta}^\top\bm{g}_{i}^{*}(\bm{\beta}_{0})     }{1+\rho\bm{\theta}^\top\bm{g}_{i}^{*}(\bm{\beta}_{0})}\Big\}\bigg|\\
			\geq& \frac{\rho}{n}\bm{\theta}^\top\sum_{i=1}^{n}\frac{\bm{g}_{i}^{*}(\bm{\beta}_{0})\bm{g}_{i}^{*}(\bm{\beta}_{0})^\top}{1+\rho\bm{\theta}^\top\bm{g}_{i}^{*}(\bm{\beta}_{0})}\bm{\theta}-\frac{1}{n}\bigg| \sum_{j=1}^{q}\bm{e}_{j}^\top\sum_{i=1}^{n}\bm{g}_{i}^{*}(\bm{\beta}_{0})\bigg|\\
			\geq&\frac{\rho\bm{\theta}^\top\mathbf{M}_{n}\bm{\theta}}{n(1+\rho Z_{n})  }-\frac{1}{n}\bigg| \sum_{j=1}^{q}\bm{e}_{j}^\top\sum_{i=1}^{n}\bm{g}_{i}^{*}(\bm{\beta}_{0})\bigg|,
		\end{aligned}
		\label{eq:inequality}
	\end{equation}
	where $\bm{e}_{j}$ is a $q$-dimensional vector with its $j$th element being $1$ and all its other elements being $0$. Let $\sigma_{1}$ denote the smallest eigenvalue of $\mathbf{M}$.  Since $\mathbf{M}$ is positive definite, then $\sigma_{1}>0$. Under condition (R.5), $\bm{\theta}^\top\mathbf{M}_{n}\bm{\theta}/n\geq \sigma_{1}+o_{p}(1)$. Since $\sum_{i=1}^{n}\bm{g}_{i}^{*}(\bm{\beta}_{0})/\sqrt{n}\rightsquigarrow\mathcal{N}(\bm{0},\mathbf{M})$, then $| \sum_{j=1}^{q}\bm{e}_{j}^\top\sum_{i=1}^{n}\bm{g}_{i}^{*}(\bm{\beta}_{0})|/n=O_{p}(n^{-1/2})$. According to~(\ref{eq:inequality}), $\rho/(1+\rho Z_{n})=O_{p}(n^{-1/2})$.  Furthermore, since $Z_{n}=o_{p}(n^{1/2})$, we can obtain that  $\rho=\|\bm{\lambda}_{0}\|=O_{p}(n^{-1/2})$.
\end{proof}

\subsection{Proof of Theorem~\ref{theo:normality}}
\label{sss:proofofnormality}
Define the empirical log-likelihood ratio as 
\[
l_E(\bm{\beta})=\sum_{i=1}^{n}\log \{1+\bm{\lambda}(\bm{\beta})^\top\bm{g}_i^\ast(\bm{\beta}) \}.
\]
Denote $\bm{\beta}=\bm{\beta}_0+\bm{u}n^{-1/3}$, for $\bm{\beta}\in \{ \bm{\beta}\mid \|\bm{\beta}-\bm{\beta}_0\|=n^{-1/3}  \}$, where $\|\bm{u}\|=1$. Similar to the proof of Lemma \ref{lemma:A1}, under conditions (R.1)--(R.4), $\mathrm{E}\|\bm{g}_{i}^{*}(\bm{\beta})\|^{3}<\infty$. 
Similar to the proof of \cite{owen1990empirical}, when $\mathrm{E}\|\bm{g}_{i}^{*}(\bm{\beta})\|^{3}<\infty$ and $\|\bm{\beta}-\bm{\beta}_0\|\leq n^{-1/3}$, we have 
\begin{equation}
\label{eq:lambda}
\begin{aligned}
\bm{\lambda}(\beta)=&\left\{\frac{1}{n}\sum_{i=1}^{n}\bm{g}_i^\ast(\bm{\beta})\bm{g}_i^\ast(\bm{\beta})^\top \right\}^{-1}\left\{\frac{1}{n}\sum_{i=1}^{n}\bm{g}_i^{\ast}(\bm{\beta})  \right\}+o(n^{-1/3}) \quad (a.s.)\\
=&O(n^{-1/3}) \quad (a.s.),
\end{aligned}
\end{equation}
uniformly about $\theta\in \{ \bm{\beta}\mid \|\bm{\beta}-\bm{\beta}_0\|\leq n^{-1/3}  \}$.

By \eqref{eq:lambda} and the Taylor expansion, we have (uniformly for $\bm{u}$),
\[
\begin{aligned}
l_E(\bm{\beta})=&\sum_{i=1}^{n}\bm{\lambda}(\bm{\beta})^\top\bm{g}_i^\ast(\bm{\beta})-\frac{1}{2}\sum_{i=1}^{n}\{\bm{\lambda}(\bm{\beta})^\top\bm{g}_i^\ast(\bm{\beta}) \}^2+o(n^{1/3}) \quad (a.s.)\\
=&\frac{n}{2}\left\{\frac{1}{n}\sum_{i=1}^{n}\bm{g}_i^\ast(\bm{\beta}) \right\}^\top\left\{\frac{1}{n}\sum_{i=1}^n\bm{g}_i^\ast(\bm{\beta}) \bm{g}_i^\ast(\bm{\beta})^\top \right\}^{-1}\left\{\frac{1}{n}\sum_{i=1}^{n}\bm{g}_i^\ast(\bm{\beta}) \right\}\\
&+o(n^{1/3}) \quad (a.s.)\\
=&\frac{n}{2}\left\{\frac{1}{n}\sum_{i=1}^{n}\bm{g}_i^\ast(\bm{\beta}_0)+\frac{1}{n}\sum_{i=1}^{n}\frac{\partial \bm{g}^\ast_i(\bm{\beta}_0)}{\partial \bm{\beta}^\top}\bm{u}n^{-1/3} \right\}^\top\left\{\frac{1}{n}\sum_{i=1}^n\bm{g}_i^\ast(\bm{\beta}) \bm{g}_i^\ast(\bm{\beta})^\top \right\}^{-1}\\
&\times \left\{\frac{1}{n}\sum_{i=1}^{n}\bm{g}_i^\ast(\bm{\beta}_0)+\frac{1}{n}\sum_{i=1}^{n}\frac{\partial \bm{g}^\ast_i(\bm{\beta}_0)}{\partial \bm{\beta}^\top}\bm{u}n^{-1/3} \right\}+o(n^{1/3}) \quad (a.s.)\\
=&\frac{n}{2}\left[O\{n^{-1/2}(\log\log n)^{1/2} \} +\mathbf{L}\bm{u}n^{-1/3} \right]^\top \times \mathbf{M}^{-1}\\
&\times \left[O\{n^{-1/2}(\log\log n)^{1/2} \} +\mathbf{L}\bm{u}n^{-1/3} \right]+o(n^{1/3}) \quad (a.s.)\\
\geq & (c-\varepsilon)n^{1/3}, \quad a.s.,
\end{aligned}
\]
where $c-\varepsilon>0$ and $c$ is the smallest eigenvalue of $\mathbf{L}^\top\mathbf{M}^{-1}\mathbf{L}$.

Similarly,
\[
\begin{aligned}
l_E(\bm{\beta}_0)=&\frac{n}{2}\left\{\frac{1}{n}\sum_{i=1}^{n}\bm{g}^\ast_i(\bm{\beta}_0) \right\}^\top\left\{\frac{1}{n}\sum_{i=1}^{n}\bm{g}^\ast_i(\bm{\beta}_0) \bm{g}^\ast_i(\bm{\beta}_0)^\ast  \right\}^{-1}\\
&\times \left\{\frac{1}{n}\sum_{i=1}^{n}\bm{g}^\ast_i(\bm{\beta}_0) \right\}+o(1) \quad (a.s.)\\
=& O(\log \log n), \quad (a.s.).
\end{aligned}
\]

Since $l_E(\bm{\beta})$ is a continuous function about $\bm{\beta}$ as $\bm{\beta}$ belongs to the ball $\|\bm{\beta}-\bm{\beta}_0\|\leq n^{-1/3}$, then, as $n\rightarrow \infty$,  $l_E(\bm{\beta})$ attains its minimum value at some point $\hat{\bm{\beta}}$ in the interior of this ball with probability 1.
%By Lemma 1 in  \cite{Qin1994Empirical}, under conditions (R.2)--(R.5), as $n\rightarrow\infty$, $R(\bm{\beta})$ attains its maximum value at some point $\hat{\bm{\beta}}$ in the interior of the ball $\|\bm{\beta}-\bm{\beta}_{0}\|\leq n^{-1/3}$  with probability $1$. 
Besides, $\hat{\bm{\beta}}$ and $\hat{\bm{\lambda}}=\bm{\lambda}(\hat{\bm{\beta}})$ satisfy $\bm{Q}_{1n}(\hat{\bm{\beta}},\hat{\bm{\lambda}})=\bm{0}$ and $\bm{Q}_{2n}(\hat{\bm{\beta}},\hat{\bm{\lambda}})=\bm{0}$, where
\[
\bm{Q}_{1n}(\bm{\beta},\bm{\lambda})=\frac{1}{n}\sum_{i=1}^{n}\frac{\bm{g}_{i}^{*}(\bm{\beta})}{1+\bm{\lambda}^\top\bm{g}_{i}^{*}(\bm{\beta})},			
\]
and
\[
\bm{Q}_{2n}(\bm{\beta},\bm{\lambda})=\frac{1}{n}\sum_{i=1}^{n}\frac{1}{1+\bm{\lambda}^\top\bm{g}_{i}^{*}(\bm{\beta})}\left(\frac{\partial \bm{g}_{i}^{*}(\bm{\beta}) }{\partial \bm{\beta}^\top}\right)^\top\bm{\lambda}.		
\]
Expanding $\bm{Q}_{1n}(\hat{\bm{\beta}},\hat{\bm{\lambda}})$ and $\bm{Q}_{2n}(\hat{\bm{\beta}},\hat{\bm{\lambda}})$ at $(\bm{\beta}_{0},\bm{0})$, by conditions (R.2)--(R.5), we have
\[
\bm{0}=\bm{Q}_{1n}(\hat{\bm{\beta}},\hat{\bm{\lambda}})=\bm{Q}_{1n}(\bm{\beta}_{0},\bm{0})+\frac{\partial \bm{Q}_{1n}(\bm{\beta}_{0},\bm{0}) }{\partial \bm{\beta}^\top}(\hat{\bm{\beta}}-\bm{\beta}_{0})+\frac{\partial \bm{Q}_{1n}(\bm{\beta}_{0},\bm{0}) }{\partial \bm{\lambda}^\top}(\hat{\bm{\lambda}}-\bm{0})+o_{p}(\delta_{n}),
\]
and
\[
\bm{0}=\bm{Q}_{2n}(\hat{\bm{\beta}},\hat{\bm{\lambda}})=\bm{Q}_{2n}(\bm{\beta}_{0},\bm{0})+\frac{\partial \bm{Q}_{2n}(\bm{\beta}_{0},\bm{0}) }{\partial \bm{\beta}^\top}(\hat{\bm{\beta}}-\bm{\beta}_{0})+\frac{\partial \bm{Q}_{2n}(\bm{\beta}_{0},\bm{0}) }{\partial \bm{\lambda}^\top}(\hat{\bm{\lambda}}-\bm{0})+o_{p}(\delta_{n}),
\]
where $\delta_{n}=\|\hat{\bm{\beta}}-\bm{\beta}_{0}\|+\|\hat{\bm{\lambda}} \|$. Therefore,
\begin{equation}
	\begin{pmatrix}\hat{\bm{\lambda}}\\\hat{\bm{\beta}}-\bm{\beta}_{0}\end{pmatrix}
	=\mathbf{S}_{n}^{-1}\begin{pmatrix}-\bm{Q}_{1n}(\bm{\beta}_{0},\bm{0})+o_{p}(\delta_{n})\\ o_{p}(\delta_{n})\end{pmatrix},
	\label{eq:expansion}
\end{equation}
where
\begin{equation}
	\mathbf{S}_{n}=\begin{pmatrix} \frac{\partial \bm{Q}_{1n}(\bm{\beta}_{0},\bm{0})}{\partial \bm{\lambda}^\top} &  \frac{\partial \bm{Q}_{1n}(\bm{\beta}_{0},\bm{0})}{\partial \bm{\beta}^\top}    \\ \frac{\partial \bm{Q}_{2n}(\bm{\beta}_{0},\bm{0})}{\partial \bm{\lambda}^\top}&\bm{0} \end{pmatrix}\stackrel{P}{\rightarrow}\mathbf{S}=\begin{pmatrix} -\mathbf{M} & \mathbf{L} \\  \mathbf{L}^\top & \bm{0}\end{pmatrix}.
	\label{eq:expansion2}
\end{equation}
From Lemma~\ref{lemma:A1}, we have $\bm{Q}_{1n}(\bm{\beta}_{0},\bm{0})=\sum_{i=1}^{n}\bm{g}_{i}^{*}(\bm{\beta}_{0})/n=O_{p}(n^{-1/2})$. Besides, from~(\ref{eq:expansion}) and~(\ref{eq:expansion2}), we know that $\delta_{n}=O_{p}(n^{-1/2})$ and
\[
\sqrt{n}(\hat{\bm{\beta}}-\bm{\beta}_{0})=( \mathbf{L}^\top \mathbf{M} \mathbf{L})^{-1} \mathbf{L}^\top \mathbf{M}^{-1}\sqrt{n}\bm{Q}_{1n}(\bm{\beta}_{0},\bm{0})+o_{p}(1).
\]
Furthermore, from Lemma~\ref{lemma:A1}, we have	$\sqrt{n}(\hat{\bm{\beta}}-\bm{\beta}_{0})\rightsquigarrow \mathcal{N}\big(\bm{0}, (\mathbf{L}^\top\mathbf{M}^{-1}\mathbf{L})^{-1}\big)$.

\subsection{Proof of Theorem~\ref{theo:chisquaretrue}}
\label{sss:proofofchisquare}
Applying  the Taylor  expansion to the formula~(\ref{eq:object}), we have
\[
-2\log R(\bm{\beta}_{0})=2\sum_{i=1}^{n}\Big[\bm{\lambda}_{0}^\top\bm{g}_{i}^{*}(\bm{\beta}_{0})-\frac{1}{2}\big\{\bm{\lambda}_{0}^\top\bm{g}_{i}^{*}(\bm{\beta}_{0})\big\}^{2} \Big]+\epsilon_{n},
\]
where $\epsilon_{n}\leq c\sum_{i=1}^{n}|\bm{\lambda}_{0}^\top\bm{g}_{i}^{*}(\bm{\beta}_{0})|^3$ in probability for some bounded positive constant $c$.
Because of the fact that $\sum_{i=1}^{n}\|\bm{g}_{i}^{*}(\bm{\beta}_{0})\|^{3}/n=o_{p}(n^{1/2})$ and  $\|\bm{\lambda}_{0}\|=O_{p}(n^{-1/2})$, we have
\[
\epsilon_{n}\leq c\sum_{i=1}^{n}|\bm{\lambda}_{0}^\top\bm{g}_{i}^{*}(\bm{\beta}_{0})|^3=O_{p}(n^{-3/2})o_{p}(n^{3/2})=o_{p}(1).
\]
So
\begin{equation}
	-2\log R(\bm{\beta}_{0})=2\sum_{i=1}^{n}\Big[\bm{\lambda}_{0}^\top\bm{g}_{i}^{*}(\bm{\beta}_{0})-\frac{1}{2}\big\{\bm{\lambda}_{0}^\top\bm{g}_{i}^{*}(\bm{\beta}_{0})\big\}^{2} \Big]+o_{p}(1).
	\label{eq:taylorexpan}
\end{equation}
Note that
\begin{equation}
	\begin{aligned}
		\bm{0}=&\frac{1}{n}\sum_{i=1}^{n}\frac{\bm{g}_{i}^{*}(\bm{\beta}_{0})}{1+\bm{\lambda}_{0}^\top\bm{g}_{i}^{*}(\bm{\beta}_{0})}\\
		=&\frac{1}{n}\sum_{i=1}^{n}\bm{g}_{i}^{*}(\bm{\beta}_{0})\bigg[1-\bm{\lambda}_{0}^\top\bm{g}_{i}^{*}(\bm{\beta}_{0})+\frac{\big\{\bm{\lambda}_{0}^\top\bm{g}_{i}^{*}(\bm{\beta}_{0})\big\}^{2}}{1+\bm{\lambda}_{0}^\top\bm{g}_{i}^{*}(\bm{\beta}_{0})}\bigg]\\
		=&\frac{1}{n}\sum_{i=1}^{n}\bm{g}_{i}^{*}(\bm{\beta}_{0})-\frac{1}{n}\sum_{i=1}^{n}\bm{g}_{i}^{*}(\bm{\beta}_{0})\bm{g}_{i}^{*}(\bm{\beta}_{0})^\top\bm{\lambda}_{0}+\frac{1}{n}\sum_{i=1}^{n}\bm{g}_{i}^{*}(\bm{\beta}_{0})\frac{\big\{\bm{\lambda}_{0}^\top\bm{g}_{i}^{*}(\bm{\beta}_{0})\big\}^{2}}{1+\bm{\lambda}_{0}^\top\bm{g}_{i}^{*}(\bm{\beta}_{0})}.
	\end{aligned}
	\label{eq:equationexpan}
\end{equation}
By Lemma~\ref{lemma:A2},
\begin{equation}
	\begin{aligned}
		\frac{1}{n}\sum_{i=1}^{n}\bigg\|\bm{g}_{i}^{*}(\bm{\beta}_{0})\frac{\big\{\bm{\lambda}_{0}^\top\bm{g}_{i}^{*}(\bm{\beta}_{0})\big\}^{2}}{1+\bm{\lambda}_{0}^\top\bm{g}_{i}^{*}(\bm{\beta}_{0})}\bigg\|&\leq \frac{1}{n}\sum_{i=1}^{n}\|\bm{g}_{i}^{*}(\bm{\beta}_{0})\|^{3}\|\bm{\lambda}_{0}\|^{2}|1+\bm{\lambda}_{0}^\top\bm{g}_{i}^{*}(\bm{\beta}_{0})|^{-1}\\
		&=o_{p}(n^{1/2})O_{p}(n^{-1})O_{p}(1)=o_{p}(n^{-1/2}).
	\end{aligned}
	\label{eq:upperbound}
\end{equation}
Therefore, from~(\ref{eq:equationexpan}),~(\ref{eq:upperbound}), and condition (R.5), we can obtain that
\[
\bm{\lambda}_{0}=\bigg\{\sum_{i=1}^{n}\bm{g}_{i}^{*}(\bm{\beta}_{0})\bm{g}_{i}^{*}(\bm{\beta}_{0})^\top\bigg\}^{-1}\sum_{i=1}^{n}\bm{g}_{i}^{*}(\bm{\beta}_{0})+o_{p}(n^{-1/2}).
\]
Substituting the value of $\bm{\lambda}_{0}$ into~(\ref{eq:taylorexpan}), we have
\begin{multline}
	-2\log R(\bm{\beta}_{0})=\sum_{i=1}^{n}\bm{\lambda}_{0}^\top\bm{g}_{i}^{*}(\bm{\beta}_{0})+o_{p}(1)\\
	=\left\{\frac{\sum_{i=1}^{n}\bm{g}_{i}^{*}(\bm{\beta}_{0})}{\sqrt{n}}\right\}^\top\left\{\frac{\sum_{i=1}^{n}\bm{g}_{i}^{*}(\bm{\beta}_{0})\bm{g}_{i}^{*}(\bm{\beta}_{0})^\top}{n}\right\}^{-1}\left\{\frac{\sum_{i=1}^{n}\bm{g}_{i}^{*}(\bm{\beta}_{0})}{\sqrt{n}}\right\}+o_{p}(1)\\
	=\left\{\mathbf{M}^{-1/2}\frac{\sum_{i=1}^{n}\bm{g}_{i}^{*}(\bm{\beta}_{0})}{\sqrt{n}}\right\}^\top\left(\mathbf{M}^{-1/2}\mathbf{M}\mathbf{M}^{-1/2} \right)^{-1}\left\{\mathbf{M}^{-1/2}\frac{\sum_{i=1}^{n}\bm{g}_{i}^{*}(\bm{\beta}_{0})}{\sqrt{n}}\right\}+o_{p}(1).
	\label{eq:trueexpan}
\end{multline}
Since $\mathbf{M}^{-1/2}\sum_{i=1}^{n}\bm{g}_{i}^{*}(\bm{\beta}_{0})/\sqrt{n}\rightsquigarrow\mathcal{N}(\bm{0},\mathbf{I}_{q\times q})$ and the rank of $\mathbf{M}$ is $q$, we have $-2\log R(\bm{\beta}_{0})\rightsquigarrow \chi^{2}(q)$.

\subsection{Proof of Theorem~\ref{theorem:fulltest}}
\label{sss:proofoffulltest}
According to~(\ref{eq:object}), the  empirical likelihood ratio test statistic is
\[
W_{1}(\bm{\beta}_{0})=2\left[\sum_{i=1}^{n}\log\left\{1+\bm{\lambda}_{0}^\top\bm{g}_{i}^{*}(\bm{\beta}_{0})\right\}-\sum_{i=1}^{n}\log\left\{1+\hat{\bm{\lambda}}^\top\bm{g}_{i}^{*}(\hat{\bm{\beta}})\right\}\right].
\]
By~(\ref{eq:object}),~(\ref{eq:trueexpan}), and the definition of $\bm{Q}_{1n}(\bm{\beta},\bm{\lambda})$, we know that
\[
\sum_{i=1}^{n}\log\left\{1+\bm{\lambda}_{0}^\top\bm{g}_{i}^{*}(\bm{\beta}_{0})\right\}=\frac{n}{2}\bm{Q}_{1n}(\bm{\beta}_{0},\bm{0})^\top\mathbf{M}^{-1}\bm{Q}_{1n}(\bm{\beta}_{0},\bm{0})+o_{p}(1).
\]
By the Taylor  expansion, we can obtain that
\[
\sum_{i=1}^{n}\log\big\{1+\hat{\bm{\lambda}}^\top\bm{g}_{i}^{*}(\hat{\bm{\beta}})\big\}=\frac{n}{2}\bm{Q}_{1n}(\bm{\beta}_{0},\bm{0})^\top\mathbf{V}\bm{Q}_{1n}(\bm{\beta}_{0},\bm{0})+o_{p}(1),
\]
where $\mathbf{V}=\mathbf{M}^{-1}\big\{\mathbf{I}_{q\times q}-\mathbf{L}(\mathbf{L}^\top\mathbf{M}^{-1}\mathbf{L})^{-1}\mathbf{L}^\top\mathbf{M}^{-1}\big\}$.
Therefore,
\[
\begin{aligned}
W_{1}(\bm{\beta}_{0})=&n\bm{Q}_{1n}(\bm{\beta}_{0},\bm{0})^\top(\mathbf{M}^{-1}-\mathbf{V})\bm{Q}_{1n}(\bm{\beta}_{0},\bm{0})+o_{p}(1)\\
=&n\bm{Q}_{1n}(\bm{\beta}_{0},\bm{0})^\top\mathbf{M}^{-1}\mathbf{L}(\mathbf{L}^\top\mathbf{M}^{-1}\mathbf{L})^{-1}\mathbf{L}^\top\mathbf{M}^{-1}\bm{Q}_{1n}(\bm{\beta}_{0},\bm{0})+o_{p}(1)\\
=&\left\{\mathbf{M}^{-1/2}\sqrt{n}\bm{Q}_{1n}(\bm{\beta}_{0},\bm{0})\right\}^\top\left\{\mathbf{M}^{-1/2}\mathbf{L}(\mathbf{L}^\top\mathbf{M}^{-1}\mathbf{L})^{-1}\mathbf{L}^\top\mathbf{M}^{-1/2}       \right\}\\
&\times  \left\{\mathbf{M}^{-1/2}\sqrt{n}\bm{Q}_{1n}(\bm{\beta}_{0},\bm{0})\right\}+o_{p}(1).
\end{aligned}
\]
By Lemma~\ref{lemma:A1}, $\mathbf{M}^{-1/2}\sqrt{n}\bm{Q}_{1n}(\bm{\beta}_{0},\bm{0})\rightsquigarrow\mathcal{N}(\bm{0},\mathbf{I}_{q\times q})$. By calculation, $\mathbf{M}^{-1/2}\mathbf{L}(\mathbf{L}^\top\mathbf{M}^{-1}\mathbf{L})^{-1}\mathbf{L}^\top\mathbf{M}^{-1/2}$ is a symmetric and idempotent matrix with the rank of  $p$.  Hence the  empirical likelihood ratio statistic $W_{1}(\bm{\beta}_{0})\rightsquigarrow \chi^{2}(p)$.

\subsection{Proof of Theorem~\ref{theorem:profiletest}}
\label{sss:proofofprofiletest}
It is obvious that 
\[
\frac{1}{n}\sum_{i=1}^{n}\frac{\partial \bm{g}_{i}^{*}(\bm{\beta}_{0})}{\partial \bm{\beta}^\top}=\left(\frac{1}{n}\sum_{i=1}^{n}\frac{\partial \bm{g}_{i}^{*}(\bm{\beta}_{1}^{0},\bm{\beta}_{2}^{0})}{\bm{\beta}_{1}^\top},\frac{1}{n}\sum_{i=1}^{n}\frac{\partial \bm{g}_{i}^{*}(\bm{\beta}_{1}^{0},\bm{\beta}_{2}^{0})}{\partial \bm{\beta}_{2}^\top}\right).
\] 
By condition (R.5),  $n^{-1}\sum_{i=1}^{n}\partial \bm{g}_{i}^{*}(\bm{\beta}_{0})/\partial \bm{\beta}^\top\stackrel{P}{\rightarrow}\mathbf{L}$. Thus, $n^{-1}\sum_{i=1}^{n}\partial \bm{g}_{i}^{*}(\bm{\beta}_{1}^{0},\bm{\beta}_{2}^{0})/\partial \bm{\beta}_{1}^\top \stackrel{P}{\rightarrow}\mathbf{L}_{1}$ and $n^{-1}\sum_{i=1}^{n}\partial \bm{g}_{i}^{*}(\bm{\beta}_{1}^{0},\bm{\beta}_{2}^{0})/\partial \bm{\beta}_{2}^\top\stackrel{P}{\rightarrow}\mathbf{L}_{2}$, where $\mathbf{L}_{1}$ and $\mathbf{L}_{2}$ are the corresponding components of $\mathbf{L}$.
By the Taylor expansion, we can obtain that
\[
\begin{aligned}
W_{2}=&-2\log \big\{R(\bm{\beta}_{1}^{0},\hat{\bm{\beta}}_{2}^{0})\big\}+2\log \big\{R(\hat{\bm{\beta}}_{1},\hat{\bm{\beta}}_{2})\big\}\\
=&\left\{\mathbf{M}^{-1/2}\sqrt{n}\bm{Q}_{1n}(\bm{\beta}_{0},\bm{0})\right\}^\top\mathbf{M}^{-1/2}\left\{\mathbf{L}(\mathbf{L}^\top\mathbf{M}^{-1}\mathbf{L})^{-1}\mathbf{L}^\top-\mathbf{L}_{2}(\mathbf{L}_{2}^\top\mathbf{M}^{-1}\mathbf{L}_{2})^{-1}\mathbf{L}_{2}^\top  \right\}   \\
&\times \mathbf{M}^{-1/2}\left\{\mathbf{M}^{-1/2}\sqrt{n}\bm{Q}_{1n}(\bm{\beta}_{0},\bm{0})\right\}+o_{p}(1).
\end{aligned}
\]
Note that
\[
\begin{aligned}
&\mathbf{L}(\mathbf{L}^\top\mathbf{M}^{-1}\mathbf{L})^{-1}\mathbf{L}^\top\\
=&\begin{pmatrix}\mathbf{L}_{1}&\mathbf{L}_{2} \end{pmatrix}\left\{\begin{pmatrix}\mathbf{L}_{1}^\top\\ \mathbf{L}_{2}^\top \end{pmatrix}\mathbf{M}^{-1}\begin{pmatrix}\mathbf{L}_{1}&\mathbf{L}_{2} \end{pmatrix} \right\}^{-1} \begin{pmatrix}\mathbf{L}_{1}^\top\\ \mathbf{L}_{2}^\top \end{pmatrix}\\
\geq &\begin{pmatrix}\mathbf{L}_{1}&\mathbf{L}_{2} \end{pmatrix}\begin{pmatrix} \bm{0}&\bm{0}\\ \bm{0}&\big(\mathbf{L}_{2}^\top\mathbf{M}^{-1} \mathbf{L}_{2}\big) ^{-1}\end{pmatrix} \begin{pmatrix}\mathbf{L}_{1}^\top\\ \mathbf{L}_{2}^\top \end{pmatrix}\\
=& \mathbf{L}_{2}\big(\mathbf{L}_{2}^\top\mathbf{M}^{-1} \mathbf{L}_{2}\big) ^{-1}\mathbf{L}_{2}^\top.
\end{aligned}
\]
Therefore, $\mathbf{L}(\mathbf{L}^\top\mathbf{M}^{-1}\mathbf{L})^{-1}\mathbf{L}^\top-\mathbf{L}_{2}(\mathbf{L}_{2}^\top\mathbf{M}^{-1}\mathbf{L}_{2})^{-1}\mathbf{L}_{2}^\top$ is non-negative definite. Since the matrices of
$\mathbf{M}^{-1/2}\mathbf{L}(\mathbf{L}^\top\mathbf{M}^{-1}\mathbf{L})^{-1}\mathbf{L}^\top\mathbf{M}^{-1/2}$ and $\mathbf{M}^{-1/2}\mathbf{L}_{2}(\mathbf{L}_{2}^\top\mathbf{M}^{-1}\mathbf{L}_{2})^{-1}\times\mathbf{L}_{2}^\top\mathbf{M}^{-1/2}$ are symmetric and idempotent,  with the ranks of  $p$ and $p-r$,  respectively,  then the  empirical likelihood ratio statistic $W_{2}\rightsquigarrow \chi^{2}(r)$.

\end{document}